\documentclass[a4paper, aps, pra, reprint, nofootinbib, superscriptaddress]{revtex4-1}
\usepackage[cm]{fullpage}
\usepackage[T1]{fontenc}
\usepackage{amssymb, amsmath, amsthm}
\makeatletter
\def\amsbb{\use@mathgroup \M@U \symAMSb}
\makeatother
\usepackage{bbold}
\usepackage{tikz}
\usepackage{verbatim}
\usepackage{subcaption}
\newcommand{\changes}[1]{#1}
\newcommand{\nbox}[2][9]{\hspace{#1pt} \mbox{#2} \hspace{#1pt}}
\newtheorem{thm}{Theorem}[section]
\newtheorem{lemma}[thm]{Lemma}
\newtheorem{prop}{Proposition}
\newtheorem{conj}{Conjecture}
\newcommand{\targetstate}{\Phi_{A'B'}}
\newcommand{\targetstatemix}{\sigma_{A'B'}}
\newcommand{\inputstate}{\rho_{AB}}
\newcommand{\auxstate}{\sigma_{A''B''}}
\newcommand{\rhoXYAB}{\rho_{XY\!AB}}

\newcommand{\pifour}{\pi/4}
\newcommand{\pitwo}{\pi/2}
\newcommand{\X}{\mathsf{X}}
\newcommand{\Y}{\mathsf{Y}}
\newcommand{\Z}{\mathsf{Z}}
\newcommand{\q}{\sqrt{2}}
\newcommand{\unit}{\mathbb{1}}
\newcommand{\meth}{STOPI}
\newcommand{\p}{\textnormal{P}}
\newcommand{\betaCHSH}{\beta_{\textnormal{CHSH}}^{*}}
\newcommand{\thrviol}{2.0014}
\newcommand{\lmax}{\lambda_{\textnormal{max}}}
\newcommand{\lmin}{\lambda_{\textnormal{min}}}
\DeclareMathOperator{\rk}{rk}
\DeclareMathOperator{\spec}{spec}

\DeclareMathOperator{\tr}{tr}

\DeclareMathOperator{\id}{id}

\def \diracspacing {0.7pt}
 
\newcommand{\ket}[1]{| \hspace{\diracspacing} #1 \rangle} 
\newcommand{\ketbra}[2]{| \hspace{\diracspacing} #1 \rangle \langle #2 \hspace{\diracspacing} |} 
\newcommand{\ketbraq}[1]{\ketbra{#1}{#1}} 

\newcommand{\tran}[0]{^\textnormal{\tiny{T}}}
\newcommand{\norm}[2][]{#1| \! #1| #2 #1| \! #1|}
\newcommand{\ave}[2][]{#1\langle #2 #1\rangle}
\newcommand{\abs}[2][]{#1| #2 #1|}

\newcommand{\cL}{\mathcal{L}}

\newcommand{\cS}{\mathcal{S}}

\newcommand{\sH}{\mathcal{H}}

\newcommand{\extractability}
{
\Xi( \inputstate \rightarrow \targetstate ) := \max_{ \Lambda_{A}, \Lambda_{B} } F \big( ( \Lambda_{A} \otimes \Lambda_{B} ) ( \inputstate ) , \targetstate \big)
}
\newcommand{\Fiso}{F_{\textnormal{iso}}}
\newcommand{\FMY}{F_{\textnormal{MY}}}
\newcommand{\salphadef}
{
\frac{ \big( \sqrt{ 8 + 2 \alpha^{2} } + 2 + \alpha \big) \big( 3 \sqrt{ 8 + 2 \alpha^{2} } - \sqrt{ 4 - \alpha^{2} } - \alpha \sqrt{ 2 } \big) }{4 ( 2 - \alpha )^2 \sqrt{ 8 + 2 \alpha^{2} } }
}
\newcommand{\balphastar}
{
b_{\alpha}^{*} := \arcsin \bigg( \sqrt{ \frac{4 - \alpha^{2}}{8} } \bigg)
}
\newcommand{\gfundef}
{
g(x) := ( 1 + \sqrt{2} ) ( \sin x + \cos x - 1)
}
\newcommand{\params}
{
\begin{align*}
s_{\alpha} &:= \salphadef,\\
\mu_{\alpha} &:= 1 - s_{\alpha} \cdot \sqrt{ 8 + 2\alpha^2 }.
\end{align*}
}
\newcommand{\alpharan}
{
\alpha \in [0, 2)
}
\newcommand{\ewav}{\varepsilon_{\textnormal{wav}}}
\newcommand{\defcq}[3]
{
#1 = \sum_{x, y=0}^{2} #2 \ketbraq{x}_{X} \otimes \ketbraq{y}_{Y} \otimes #3^{xy}
}
\newcommand{\finexp}[1]
{
4 \big[ 8 \varepsilon_{00} + 16 (\varepsilon_{02} + \varepsilon_{20} + \varepsilon_{22}) + 3 \varepsilon_{#1} \big]
}
\newcommand{\finexpu}
{
\frac{1}{62} \big[ 8 \varepsilon_{00} + 16 (\varepsilon_{02} + \varepsilon_{20} + \varepsilon_{22}) + 3 \varepsilon_{01} + 3 \varepsilon_{10} \big]
}
\newtheorem*{prop:counterexample-1}{Proposition \ref*{prop:ewav-eigenvalue}}
\newtheorem*{prop:counterexample-2}{Proposition \ref*{prop:centre-bound}}
\newcommand{\propewaveigenvalue}
{
Let
\begin{equation*}
\rho_{AB}^{xy} =
\begin{cases}
\ketbraq{11} &\nbox{if} (x, y) = (0, 0),\\
\frac{1}{2} \big( \ketbraq{00} + \ketbraq{11} \big) &\nbox{if} (x, y)
= (0, 1), (0, 2), (1, 0), (2, 0), \\
\frac{1}{2} \big( \ketbraq{01} + \ketbraq{10} \big) &\nbox{if} (x, y) = (2, 2).
\end{cases}
\end{equation*}
For these six points for fixed extraction channels $\Lambda_{A}^{x}$ and $\Lambda_{B}^{y}$ define
\begin{equation}
\label{eq:epsilonxy}
\varepsilon_{xy} := \frac{1}{2} - \ave[\big]{ ( \Lambda_{A}^{x} \otimes \Lambda_{B}^{y} )( \rho_{AB}^{xy} ), \targetstate^{+} }.
\end{equation}
Note that $\varepsilon_{xy} \geq 0$, since the states $\rho_{AB}^{xy}$
are separable. If
\begin{equation*}
\omega_{A} := \Lambda_{A}^{1} \bigg( \frac{ \unit_2 }{2} \bigg) \nbox{and} \omega_{B} := \Lambda_{B}^{1} \bigg( \frac{ \unit_2 }{2} \bigg),
\end{equation*}
then
\begin{align*}
\lmin( \omega_{A} ) &\leq \finexp{10},\\
\lmin( \omega_{B} ) &\leq \finexp{01}.
\end{align*}
In particular, we have $\lmin( \omega_{A} ), \lmin( \omega_{B} ) \leq 248 \ewav$ for
\begin{equation*}
\ewav := \finexpu.
\end{equation*}
}
\newcommand{\propcentrebound}
{
Let $\Lambda_{A}, \Lambda_{B}$ be qubit channels such that the smaller eigenvalues of the normalised qubit density matrices
\begin{equation*}
\Lambda_{A} \bigg( \frac{ \unit_2 }{2} \bigg) \nbox{and} \Lambda_{B} \bigg( \frac{ \unit_2 }{2} \bigg)
\end{equation*}
are at most $\lambda$. Then, for any pair of  maximally entangled two-qubit states $\Psi_{1}, \Psi_{2}$ we have
\begin{equation*}
\ave[\big]{ ( \Lambda_{A} \otimes \Lambda_{B} )( \Psi_{1} ), \Psi_{2} } \leq \frac{1}{2} + 2 \lambda.
\end{equation*}
}
\usepackage{xcolor}
\definecolor{darkred}{RGB}{130, 0, 0}
\definecolor{darkgreen}{RGB}{10, 70, 10}
\definecolor{darkblue}{RGB}{0, 0, 130}
\usepackage{hyperref}
\hypersetup{colorlinks, breaklinks, linkcolor=darkred, urlcolor=darkblue, citecolor=darkgreen}
\begin{document}
\title{Robust self-testing of two-qubit states}
\author{Tim Coopmans}
\affiliation{QuTech, Delft University of Technology, Lorentzweg 1, 2628 CJ Delft, The Netherlands}
\author{J\k{e}drzej Kaniewski}
\affiliation{Center for Theoretical Physics, Polish Academy of Sciences, Al.~Lotnik{\'o}w 32/46, 02-668 Warsaw, Poland}
\author{Christian Schaffner}
\affiliation{QuSoft, University of Amsterdam, Science Park 123, 1098 XG Amsterdam, The Netherlands}
\date{\today}
\begin{abstract}
It is well-known that observing nonlocal correlations allows us to draw conclusions about the quantum systems under consideration. In some cases this yields a characterisation which is essentially complete, a phenomenon known as self-testing. Self-testing becomes particularly interesting if we can make the statement robust, so that it can be applied to a real experimental setup. For the simplest self-testing scenarios the most robust bounds come from the method based on operator inequalities. In this work we elaborate on this idea and apply it to the family of tilted CHSH inequalities. These inequalities are maximally violated by partially entangled two-qubit states and our goal is to estimate the quality of the state based only on the observed violation. For these inequalities we have reached a candidate bound and while we have not been able to prove it analytically, we have gathered convincing numerical evidence that it holds. Our final contribution is a proof that in the usual formulation, the CHSH inequality only becomes a self-test when the violation exceeds a certain threshold. This shows that self-testing scenarios fall into two distinct classes depending on whether they exhibit such a threshold or not.
\end{abstract}
\maketitle
\section{Introduction}
Among the many sins of quantum mechanics, correlations between space-like separated systems occupy a rather special place. Stronger-than-classical correlations~\cite{einstein35a, bell64a} were initially seen as a problem, but have now become an inherent (and useful) feature of the quantum world. Investigating the difference between correlations achievable in quantum mechanics and in classical (local-realistic) theories goes under the name of Bell nonlocality~\cite{brunner14a}, and one of the great achievements of this field is the ability to rule out any classical description of the system under consideration based \emph{only} on the observed statistics. While clearly of fundamental importance, it turns out that this argument can be pushed one step further.

If we can rule out a classical description, our next guess is that the system is governed by quantum mechanics. Under this assumption it makes sense to ask which features of the quantum system give rise to such strikingly non-classical behaviour. Can we, for instance, deduce something about the quantum state or the measurements performed?

While it is clear that in order to observe nonlocal correlations one must perform incompatible measurements on entangled quantum systems, it is not clear which meaningful quantitative statements can be made. It might, therefore, come as a surprise that certain nonlocal correlations can be realised in an essentially unique manner. While this observation can be found in the early works of Tsirelson~\cite{tsirelson87a, tsirelson93a}, Summers and Werner~\cite{summers87a} and Popescu and Rohrlich~\cite{popescu92a}, it did not attract much attention until the seminal work of Mayers and Yao~\cite{mayers98a, mayers04a}. Mayers and Yao realised that this effect can be used to \emph{certify} quantum devices under minimal assumptions and they called this phenomenon \emph{self-testing}.

The goal of self-testing is to make quantitative statements about the \emph{quantum realisation}, e.g.~about the entanglement present in a quantum state or about the incompatibility of the measurements performed. Self-testing is closely related to the field of \emph{device-independent cryptography} whose goal is to certify properties of the classical output produced by quantum devices. Device-independent cryptography is a promising solution for randomness generation~\cite{colbeck06a, pironio10a, colbeck11a, vazirani12a, miller16b, bouda14a}, quantum key distribution~\cite{barrett05b, acin06a, acin07a, reichardt13a, vazirani14a, arnonfriedman18a} and several other tasks~\cite{silman11a, kaniewski16a, ribeiro18a, ribeiro16a, ribeiro18b}. For a brief overview of device-independent cryptography, we recommend Ref.~\cite{ekert14a} (focus on quantum key distribution) and Ref.~\cite{acin16a} (focus on randomness generation). For a comprehensive review on both philosophical and technological aspects of randomness in quantum physics, we refer the reader to Ref.~\cite{bera17a}.

In this work, we focus solely on the task of self-testing in its most common formulation, i.e.~when the goal is to certify the state and the measurements performed on it.\footnote{Note that other quantum objects such as quantum channels~\cite{sekatski18a}, entangled measurements~\cite{bancal18a, renou18a} or weak measurements~\cite{wagner18a} can be self-tested in more complex scenarios.} While there is a large class of scenarios in which self-testing statements have been proven, most results only apply if the observed statistics are (almost) perfect~\cite{bardyn09a, mckague14a, mckague12a, yang13a, bamps15a, mckague16a, wang16a, supic16a, mckague17a, coladangelo17a, kalev17a, andersson17a, supic18a, coladangelo17c}. While such results are \emph{robust} in the sense that they are stable under sufficiently small perturbations, the obtained noise tolerance is not relevant from the experimental point of view. Our goal, on the other hand, is to derive self-testing statements which can be applied to real statistics collected in real experiments.\footnote{For an intuitive explanation of the difference between robustness and experimentally-relevant robustness see Section I of Ref.~\cite{kaniewski17a}.} Such results are of interest to both experimentalists~\cite{tan17a, zhang19a, zhang18a} and theoreticians investigating specific physical setups~\cite{lee17a}, but deriving them turns out to be significantly more challenging.

The first result of this type is due to Bardyn et al.~\cite{bardyn09a} and there are currently two methods of deriving such results: the ``swap method''~\cite{bancal15a, yang14a, pal14a} and the ``self-testing from operator inequalities ({\meth}) method''~\cite{kaniewski16b}. While the swap method is extremely versatile and can be (at least in principle) immediately applied to any Bell scenario, it has two weaknesses. First of all, it is a numerical method which scales unfavourably with the dimension of the system we wish to certify: the largest states certified using this method until today consist of two ququarts~\cite{wu16a} or four qubits~\cite{pal14a}. The second, and more severe, disadvantage of the swap method is that the output of the computation is just a number, which gives little intuition about the underlying physics.

The {\meth} method, on the other hand, is more time-consuming, as it requires a more thorough understanding of the particular self-testing scenario, but the resulting bounds are significantly stronger (in some cases even tight). In Ref.~\cite{kaniewski16b} the {\meth} method was used to derive analytic self-testing bounds for the CHSH~\cite{clauser69a} and Mermin~\cite{mermin90a} inequalities. In this work we applied this method to self-test partially entangled pure two-qubit states using the family of tilted CHSH inequalities. Investigating some special cases led us to conjecture a particular form of the self-testing statement. While we were not able to prove it analytically, we have gathered strongly convincing numerical evidence that it holds. The conjectured statement improves on the bounds obtained from the swap method~\cite{bancal15a}.

Our second contribution is a proof that the CHSH inequality becomes a self-test only above a certain violation. More specifically, we have constructed a state which violates the CHSH inequality, but does not satisfy the usual self-testing criteria. This is in contrast with the Mermin inequality in which the value of the self-testing threshold coincides with the maximum value achievable if only two out of three parties are entangled.

\medskip
In Section~\ref{sec:preliminaries} we formalise the problem of self-testing, while in Section~\ref{sec:stopi} we explain the {\meth} method. In Section~\ref{sec:tilted-chsh} we present the conjectured robust self-testing bounds for partially entangled two-qubit states. In Section~\ref{sec:nontrivial-threshold} we explain the construction of the state that violates the CHSH inequality, but for which none of the usual self-testing statements can be made. In Section~\ref{sec:conclusions} we summarise our results and discuss some open problems.
\section{Preliminaries}
\label{sec:preliminaries}
In this section we establish the basic notation and formalise the problem of self-testing.
\subsection{Notation}
We denote the identity matrix by $\unit$ and the Pauli matrices by $\X, \Y, \Z$. For a Hermitian matrix $X$ we use $\lmax(X)$ and $\lmin(X)$ to denote its largest and smallest eigenvalue, respectively.

For arbitrary linear operators $X$ and $Y$ we use $\ave{X, Y} := \tr( X^{\dagger} Y )$ to denote the Hilbert-Schmidt inner product and $\norm{X}_{p}$ to denote the Schatten $p$-norm. For a positive semidefinite operator $A$, $B := \sqrt{A}$ is the unique positive semidefinite operator satisfying $B^{2} = A$. The fidelity of two positive semidefinite operators $A$ and $B$ is defined as $F(A, B) = \norm[\big]{ \sqrt{A} \sqrt{B} }_{1}^{2}$.

The Hilbert space corresponding to register $A$ is denoted by $\sH_{A}$ and in this work we assume all the Hilbert spaces to be finite-dimensional. The set of linear operators acting on $\sH$ is denoted by $\cL(\sH)$.

For a completely positive map $\Lambda : \cL(\sH_{A}) \to \cL(\sH_{B})$, the dual map $\Lambda^{\dagger} : \cL(\sH_{B}) \to \cL(\sH_{A})$ is the unique linear map which satisfies $\ave{ \Lambda(X), Y } = \ave{ X, \Lambda^{\dagger}(Y) }$ for all $X \in \cL(\sH_{A})$ and $Y \in \cL(\sH_{B})$. The map $\Lambda$ is a quantum channel if it is trace-preserving, which is equivalent to the dual map $\Lambda^{\dagger}$ being unital, i.e.~$\Lambda^{\dagger}(\unit_{B}) = \unit_{A}$.

The Choi-Jamio{\l}kowski isomorphism states that completely positive maps $\Lambda : \cL(\sH_{A}) \to \cL(\sH_{B})$ are in 1-1 correspondence with positive semidefinite operators acting on $\sH_{A} \otimes \sH_{B}$. Let $\{ \ket{j} \}_{j = 1}^{d}$ be the standard basis on $\sH_{A}$, let $\sH_{A'} \cong \sH_{A}$ and let
\begin{equation*}
\Omega_{AA'} = \ketbraq{\Omega}_{AA'} \nbox{for} \ket{\Omega}_{AA'} = \sum_{j = 1}^{d} \ket{j}_{A} \ket{j}_{A'}
\end{equation*}
be an unnormalised maximally entangled state. The \changes{(unnormalised)} Choi state of $\Lambda$, denoted by $C_{AB}$, is defined as
\begin{equation*}
C_{AB} := ( \id_{A} \otimes \Lambda_{A'} )(\Omega_{AA'})
\end{equation*}
and it is well-known that for any $X \in \cL(\sH_{A})$
\begin{equation*}
\Lambda(X) = \tr_{A} \big[ C_{AB} (X_{A}\tran \otimes \unit_{B}) \big],
\end{equation*}
where $\tran$ denotes the transpose in the standard basis. If $\Lambda$ is trace-preserving, then $C_{A} = \unit$, whereas if $\Lambda$ is unital, then $C_{B} = \unit$.
\subsection{Self-testing of quantum states}
\label{sec:formulation}
Consider the usual Bell scenario in which two space-like separated parties, Alice and Bob, perform local measurements on a shared quantum state. Alice and Bob would like to certify that the state they share is entangled, but as they do not trust their measurement devices, they are unable to perform full state tomography. Their only option is to choose measurement settings, observe the outcomes and collect statistics. To simplify the problem we assume that their devices behave in the same way every time they are used, i.e.~that they give rise to a well-defined conditional probability distribution $\Pr( a, b | x, y )$, where $a$ and $b$ are outputs and $x$ and $y$ are inputs of Alice and Bob, respectively. Since the probability vector $\p = ( \Pr( a, b | x, y ) )_{abxy}$ can be estimated to arbitrary precision and we are interested in the fundamental aspects of self-testing, we assume to have access directly to the exact probability distribution $\p$.\footnote{Not surprisingly drawing conclusions from a finite set of data is significantly harder, see Refs.~\cite{gill14a, elkouss16a, lin18a}.} 

From a mathematical point of view, self-testing of quantum states boils down to the following question:
\begin{center}
``Given a conditional probability distribution
\begin{equation*}
\p = ( \Pr( a, b | x, y ) )_{abxy}
\end{equation*}
which comes from measuring a quantum system, i.e.~
\begin{equation*}
\Pr( a, b | x, y ) = \tr \big[ ( P_{a}^{x} \otimes Q_{b}^{y} ) \inputstate \big],
\end{equation*}
what can we deduce about the unknown state $\inputstate$?''
\end{center}
We intentionally denote the unknown state by $\inputstate$, as we do not want to assume its purity.\footnote{Under the purity assumption even classical correlations are sufficient to certify entanglement~\cite{sikora16a}.} \changes{Let us also emphasise here that no knowledge of the observables is assumed, which makes self-testing a significantly different problem from quantum state tomography.}

It is important to realise that the observed statistics $\Pr( a, b | x, y )$ can never  \emph{uniquely} determine the state. Indeed, the two equivalences we must always allow for are: (i) local isometries and (ii) the presence of additional degrees of freedom. Motivated by these limitations we say that $\inputstate$ \emph{contains} $\targetstatemix$ if there exist local quantum channels $\Lambda_{A} : \cL(\sH_{A}) \to \cL(\sH_{A'})$ and $\Lambda_{B} : \cL(\sH_{B}) \to \cL(\sH_{B'})$ that \emph{extract} a perfect copy of $\targetstatemix$ from $\inputstate$, i.e.
\begin{equation}
\label{eq:perfect-extractability}
(\Lambda_A \otimes \Lambda_B) (\inputstate) = \targetstatemix.
\end{equation}
It is intuitively clear that this formulation is equivalent to the usual formulation using isometries and an auxiliary state, but for completeness we provide a proof in Appendix~\ref{app:formulation}.\footnote{At the end of Appendix~\ref{app:formulation}, we also point out that the formulation with unitaries instead of isometries is not quite correct.}

The concept of local extraction channels is well-aligned with the conditions of a Bell test in which Alice and Bob are only allowed local measurements (no communication) and they must always produce an outcome (from a fixed alphabet). Similarly, we require the extraction channels to act locally and deterministically produce a state (of the correct dimension).

Replacing local extraction channels by a distillation procedure, i.e.~allowing for classical communication, completely changes the problem. Note that the same phenomenon occurs in Bell nonlocality, where states can be preprocessed to enhance their nonlocal properties~\cite{liang06a}.

A self-testing statement consists of two components: (i) a quantum-realisable probability distribution $\p^{*}$ and (ii) a pure bipartite state $\targetstate$. The statement asserts that if an unknown state $\inputstate$ is capable of producing the probability distribution $\p^{*}$ (under some local measurements), then $\inputstate$ must contain $\targetstate$.

Of course, in a real experiment one never actually observes the exact probability distribution $\p^{*}$,\footnote{The two most obvious obstacles are experimental noise and finite statistics.} which means that a new, robust version of Eq.~\eqref{eq:perfect-extractability} is needed. For exactly that purpose the channel formulation is particularly convenient, as it is immediately clear how to turn the original requirement into an approximate statement. We define the \emph{extractability} of $\targetstate$ from $\inputstate$ as~\cite{bardyn09a, kaniewski16b}
\begin{equation}
\label{eq:extractability}
\extractability,
\end{equation}
where the maximisation is taken over all quantum channels from $A$ to $A'$ and $B$ to $B'$, respectively. It is clear that extractability is invariant under local unitaries applied to $\targetstate$, i.e.~it depends only on the Schmidt coefficients of the target state. The maximal value of extractability equals $1$ and implies that $\inputstate$ contains $\targetstate$. The lowest value, on the other hand, equals $\lambda_{0}^2$, where $\lambda_{0}$ is the largest Schmidt coefficient of $\targetstate$, because Alice and Bob can always replace $\inputstate$ with a pure product state. Moreover, extractability is convex in the input state, which implies that $\Xi( \inputstate \rightarrow \targetstate ) = \lambda_{0}^2$ whenever $\inputstate$ is separable. Note that there exist other measures for robust self-testing, but extractability is the only one for which experimentally-relevant robustness has been proven (see Appendix~\ref{app:robust-measures} for details).

In this language a self-testing statement says that if $\inputstate$ is capable of producing $\p^{*}$, then $\Xi( \inputstate \to \targetstate ) = 1$. A robust version states that observing statistics close to $\p^{*}$ implies that the extractability is close to $1$. More generally, we are interested in deriving a nontrivial lower bound on $\Xi( \inputstate \to \targetstate )$ as a function of the observed statistics.

In this work, instead of looking at the entire probability distribution $\p$, we focus on some suitably chosen Bell function. A Bell function is specified by a vector of real coefficients $( c_{abxy} )_{abxy}$, and its value evaluated on the probability distribution $\p$ equals
\begin{equation*}
\beta := \sum_{abxy} c_{abxy} \Pr( a, b | x, y ).
\end{equation*}
If $\beta_{C}$ and $\beta_{Q}$ are the maximal classical and quantum values, respectively, then our goal is to prove
\begin{equation}
\label{eq:self-testing-statement}
\Xi( \inputstate \rightarrow \targetstate ) \geq f(\beta)
\end{equation}
for some explicit function $f : [ \beta_{C}, \beta_{Q} ] \to [0, 1]$. While in principle $f$ could be an arbitrary function, we can without loss of generality assume that it is non-decreasing. Since any state capable of producing the Bell violation of $\beta$ is also capable of producing any violation in the interval $[\beta_{C}, \beta]$, we can define
\begin{equation*}
f_{\textnormal{nd}}(\beta) := \sup_{x \in [\beta_{C}, \beta]} f(x),
\end{equation*}
where the subscript in $f_{\textnormal{nd}}$ stands for \emph{non-decreasing}, and we immediately see that
\begin{equation*}
\Xi( \inputstate \rightarrow \targetstate ) \geq f_{\textnormal{nd}}(\beta).
\end{equation*}
While such trade-offs could be investigated for arbitrary combinations of target state and Bell function, the term self-testing is only used if the maximal violation of the Bell function certifies the presence of the target state, i.e.~$f( \beta_{Q} ) = 1$. A self-testing statement is called robust if $f(\beta) \to 1$ as $\beta \to \beta_{Q}$.

An important advantage of self-testing statements based only on the Bell value is the fact that we can assess their tightness by deriving an explicit upper bound on $f(\beta)$. If the Bell inequality is not violated, we cannot improve over the trivial bound of $\lambda_{0}^{2}$, i.e.~$f(\beta_{C}) = \lambda_{0}^{2}$. On the other extreme, by assumption we have $f(\beta_{Q}) = 1$. Since every intermediate violation can be achieved as a mixture of these two points, we cannot hope to certify extractability larger than the value corresponding to such a mixture. This leads to an upper bound of the form
\begin{equation}
\label{eq:upper-bound}
f(\beta) \leq \lambda_{0}^{2} + (1 - \lambda_{0}^{2}) \cdot \frac{\beta - \beta_{C}}{\beta_{Q} - \beta_{C}}.
\end{equation}
This upper bound tells us how much room for improvement there potentially is and it is worth mentioning that in some scenarios, one can prove self-testing statements matching this upper bound~\cite{kaniewski16b}. A good indication of the strength of a self-testing bound is the critical Bell value above which the statement becomes nontrivial, i.e.
\begin{equation*}
\beta^{*}_{f} := \inf_{\beta} \big\{ f(\beta) > \lambda_{0}^{2} \big\}.
\end{equation*}
Clearly, $\beta^{*}_{f}$ is computed for a specific self-testing bound (i.e.~a particular function $f$) and is not a fundamental property of the Bell inequality under consideration.
\section{Self-testing from operator inequalities}
\label{sec:stopi}
The {\meth} method was introduced and applied to two specific examples in Ref.~\cite{kaniewski16b}. Here we provide a more detailed discussion of the underlying idea.

Our goal is to prove a lower bound on the extractability as a function of the observed Bell violation $\beta$. The {\meth} method is constructive: given a quantum realisation, which consists of the shared state $\rho_{AB}$, the measurements of Alice $\{ P_{a}^{x} \}$ and the measurements of Bob $\{ Q_{b}^{y} \}$, we explicitly construct the local extraction channels $\Lambda_{A}, \Lambda_{B}$ and we provide a lower bound on their performance as a function of $\beta$. The extraction channel of Alice $\Lambda_{A} : \cL(\sH_{A}) \to \cL(\sH_{A'})$ is built out of her measurement operators $\{ P_{a}^{x} \}$ and similarly the extraction channel of Bob $\Lambda_{B} : \cL(\sH_{B}) \to \cL(\sH_{B'})$ depends only on $\{ Q_{b}^{y} \}$. We are interested in the fidelity
\begin{equation*}
F \big( ( \Lambda_{A} \otimes \Lambda_{B} ) ( \inputstate ) , \targetstate \big),
\end{equation*}
but since $\targetstate$ is a pure state, we can replace the fidelity by the inner product, which allows us to replace the channels by their duals
\begin{align*}
F \big( ( \Lambda_{A} &\otimes \Lambda_{B} ) ( \inputstate ) , \targetstate \big)\\
&= \ave{ ( \Lambda_{A} \otimes \Lambda_{B} ) ( \inputstate ) , \targetstate }\\
&= \ave{ \inputstate , ( \Lambda_{A}^{\dagger} \otimes \Lambda_{B}^{\dagger} ) (\targetstate) }.
\end{align*}
Define
\begin{equation}
\label{eq:K-operator}
K := ( \Lambda_{A}^{\dagger} \otimes \Lambda_{B}^{\dagger} ) (\targetstate)
\end{equation}
and note that this operator depends \emph{only} on the measurement operators (and not on the input state $\inputstate$). Another operator that depends only on the measurement operators is the Bell operator defined as
\begin{equation*}
W := \sum_{abxy} c_{abxy} P_{a}^{x} \otimes Q_{b}^{y},
\end{equation*}
which by construction satisfies $\tr (W \inputstate) = \beta$. We might therefore hope to prove an operator inequality of the form
\begin{equation}
\label{eq:operator-inequality}
K \geq s W + \mu \unit
\end{equation}
for suitably chosen (real) constants $s$ and $\mu$. If we prove this operator inequality \emph{for all choices of local measurements} on Alice and Bob, it immediately implies that for \emph{any} input state $\inputstate$ we have
\begin{align*}
\Xi( \inputstate \rightarrow \targetstate ) &\geq F \big( ( \Lambda_{A} \otimes \Lambda_{B} ) ( \inputstate ) , \targetstate \big)\\
&= \ave{ \inputstate , K } \geq \ave{ \inputstate , sW + \mu \unit }\\
&= s \beta + \mu.
\end{align*}
Therefore, we obtain precisely a self-testing statement of the form given in Eq.~\eqref{eq:self-testing-statement} for
\begin{equation*}
f(\beta) = s \beta + \mu.
\end{equation*}
This approach reduces the problem of self-testing to three steps: (i) constructing suitable extraction channels, (ii) choosing the right constants $s$ and $\mu$ and (iii) proving the resulting operator inequality.
\subsection{Extraction channels from measurement operators}
\label{sec:extraction-channels}
Given a set of measurements operators $\{ P_{a}^{x} \}$ acting on $\sH_{A}$ we want to construct an extraction channel $\Lambda_{A} : \cL(\sH_{A}) \to \cL(\sH_{A'})$, where the Hilbert space $\sH_{A'}$ is determined by the target state. Let us first point out that for the purpose of deriving self-testing statements it suffices to construct channels for projective measurement operators. In the case of non-projective measurement operators, Alice starts her extraction procedure by enlarging her Hilbert space until she can find projective measurements reproducing precisely the same statistics. She would then construct an extraction channel using the new, projective measurement operators.

Instead of first constructing the channel and then taking its dual, it is easier to construct the dual channel $\Lambda^{\dagger} : \cL(\sH_{A'}) \to \cL(\sH_{A})$ directly and it is convenient to specify it through its Choi state. The dual channel must be unital, so the Choi state $C_{A'A}$ must satisfy $C_{A} = \unit$. If $\{ O_{j} \}_{j} \in \cL( \sH_{A'} )$ is an operator basis on $\cL( \sH_{A'} )$, the Choi state can be written as
\begin{equation*}
C_{A'A} := \sum_{j} O_{j} \otimes F_{j} \big( \{ P_{a}^{x} \} \big)
\end{equation*}
for some collection of functions $\{ F_{j} \}$ such that $F_{j} : [\cL( \sH_{A} )]^{\times k } \to \cL( \sH_{A} )$, where $k$ is the product of the number of inputs and outputs. In principle, the only restriction on $\{ F_{j} \}$ is that the resulting operator must be a valid Choi operator \emph{for all} sets of valid measurement operators $\{ P_{a}^{x} \}$, but it is natural to choose extraction channels satisfying certain conditions.

First of all, sets of measurement operators which are related by a unitary should be treated in an equivalent manner, i.e.
\begin{equation*}
F_{j} \big( \{ U P_{a}^{x} U^{\dagger} \} \big) = U F_{j} \big( \{ P_{a}^{x} \} \big) U^{\dagger}
\end{equation*}
for all unitaries $U$ and all $j$. We call such extraction channels \emph{covariant} with respect to the unitary group.

Moreover, whenever the measurement operators exhibit a certain direct-sum structure, the extraction channels should preserve it. Given one set of measurements $\{P_{a}^{x, 0}\}$ acting on $\sH_{A_{0}}$ and another set of measurements $\{P_{a}^{x, 1}\}$ acting on $\sH_{A_{1}}$, we should have
\begin{equation*}
F_{j} \big( \{ P_{a}^{x, 0} \oplus P_{a}^{x, 1} \} \big) = F_{j} \big( \{ P_{a}^{x, 0} \} \big) \oplus F_{j} \big( \{ P_{a}^{x, 1} \} \big).
\end{equation*}
Restricting ourselves to extraction channels satisfying these two criteria makes it easier to analyse the resulting operator inequalities. As explained in the next section these restrictions do not affect the obtained bounds.

Since the target state is pure, we can assume that $\sH_{B'} \cong \sH_{A'}$ and we can choose the same operator basis for $\sH_{B'}$. Analogous to $C_{A'A}$ the Choi state describing $\Lambda_{B}^{\dagger}$ reads
\begin{equation*}
C_{B'B} := \sum_{j} O_{j} \otimes G_{j} \big( \{ Q_{b}^{y} \} \big).
\end{equation*}
Computing the $K$ operator gives
\begin{align*}
K &= ( \Lambda_{A}^{\dagger} \otimes \Lambda_{B}^{\dagger} ) (\targetstate)\\
&= \tr_{A'B'} \big[ ( C_{A'A} \otimes C_{B'B} ) ( \targetstate\tran \otimes \unit_{AB} ) \big]\\
&= \sum_{jk} \alpha_{jk} \, F_{j} \big( \{ P_{a}^{x} \} \big) \otimes G_{k} \big( \{ Q_{b}^{y} \} \big),
\end{align*}
where $\alpha_{jk} := \tr \big[ ( O_{j} \otimes O_{k} ) \targetstate\tran \big]$.
\subsection{Choosing the constants}
Since we are interested in non-decreasing functions of $\beta$, we restrict ourselves to the case $s > 0$, but otherwise all values of $s$ are in principle worth considering. For a particular choice of extraction channels and $s$, we define
\begin{equation}
\label{eq:infimum}
\mu(s) := \inf \, \lmin ( K - s W ),
\end{equation}
where the infimum is taken over all possible measurements of Alice and Bob (in all finite dimensions). Clearly, this is simply the largest value of $\mu$ for which the operator inequality~\eqref{eq:operator-inequality} holds for all possible measurements. To see that $\mu(s)$ does not diverge to $- \infty$, note that
\begin{equation*}
\mu(s) \geq \inf \, \lmin ( - s W ) \geq - \sup \norm{ s W }_{\infty} \geq - s \sum_{abxy} \abs{ c_{abxy} }.
\end{equation*}

It should now be clear why the restrictions discussed in the previous section simplify the computation of $\mu(s)$. Requiring the extraction channels to be covariant ensures that the spectrum of $K - sW$ is not affected by applying local unitaries to the measurement operators of Alice and Bob, which significantly reduces the parameter space. Requiring the channels to preserve the direct-sum structure ensures that the same direct-sum structure is inherited by the operator $K - sW$ which facilitates bounding its spectrum.

The quantity $\mu(s)$ is in general hard to compute, but if we were able to
do so for a fixed choice of extraction channels, then we would obtain a
family of lower bounds of the form
\begin{equation*}
f_{s}(\beta) = s \beta + \mu(s)
\end{equation*}
parametrised by $s > 0$.\footnote{Every $s > 0$ gives a valid bound, but for poor choices of extraction channels and/or the parameter $s$ the bound might be trivial for the entire range of $\beta \in [\beta_{C}, \beta_{Q}]$.}  All these bounds could be collected in a single function defined as
\begin{equation*}
\sup_{s > 0} \big( s \beta + \mu(s) \big).
\end{equation*}
In fact, we could also optimise over the choice of extraction channels. Such an optimisation might seem particularly advantageous as we would expect that extraction channels in the regime $\beta \approx \beta_{Q}$ should be rather different from those in the regime $\beta \approx \beta_{C}$. It is, therefore, rather surprising that in all the examples considered in Ref.~\cite{kaniewski16b} and in this work, the best lower bounds come from a single choice of extraction channels and a single value of $s$. This situation stands in contrast with the swap method in which it is beneficial to tailor the extraction channels to the observed violation (see Eqs.~(33) and (34) of Ref.~\cite{bancal15a}).

In this work we focus on the case where all the systems are finite-dimensional, but the method can be equally well applied to infinite-dimensional systems as long as the construction of extraction channels from measurement operators and the proof of the relevant operator inequality carry over to the infinite-dimensional case.
\subsection{Extracting a qubit from two binary observables}
\label{sec:extracting-qubit}
A binary measurement $\{P_{0}, P_{1}\}$ can be conveniently represented as an observable $A := P_{0} - P_{1}$ (and since $P_{0} + P_{1} = \unit$ this mapping is a bijection). An observable is a Hermitian operator $A = A^{\dagger}$ satisfying $- \unit \leq A \leq \unit$, whereas projective measurements give rise to observables satisfying $A^{2} = \unit$.

The case of two binary observables is particularly simple due to Jordan's lemma which completely characterises the interaction between two projective observables. More specifically, it states that given two projective observables $A_{0}$ and $A_{1}$ one can find a unitary which simultaneously block-diagonalises $A_{0}$ and $A_{1}$ into blocks of size $1 \times 1$ or $2 \times 2$. There are four distinct types of $1 \times 1$ blocks corresponding to $A_{0} = \pm 1, A_{1} = \pm 1$, whereas the $2 \times 2$ blocks form a 1-parameter family given by
\begin{align}
\label{eq:A0-def}
A_{0} &:= \cos a \cdot \X + \sin a \cdot \Z,\\
\label{eq:A1-def}
A_{1} &:= \cos a \cdot \X - \sin a \cdot \Z
\end{align}
for $a \in (0, \pitwo)$. In Section~\ref{sec:extraction-channels} we have argued that by enlarging the Hilbert space we can focus solely on projective measurements. Similarly, in this case we could enlarge the Hilbert space to ensure that every $1 \times 1$ block is paired up with another suitably chosen $1 \times 1$ block such that the two together are unitarily equivalent to a $2 \times 2$ block corresponding to $a = 0$ or $a = \pitwo$. As before, this grouping operation would be the first step of the extraction channel. It is not strictly necessary, but it makes the analysis easier, since it ensures that the observables are just a direct sum of $2 \times 2$ blocks parametrised by $a \in [0, \pitwo]$.

Since we restrict ourselves to covariant extraction channels, we can assume that the observables are already in block-diagonal form. Moreover, the channels respect the direct-sum structure, which implies that we only need to propose a 1-parameter family of qubit channels corresponding to the $2 \times 2$ blocks. If the extraction channels for Alice and Bob are denoted by $\Lambda_{A}(a)$ and $\Lambda_{B}(b)$, respectively, then
\begin{equation*}
K(a, b) := \big( \Lambda_{A}^{\dagger}(a) \otimes \Lambda_{B}^{\dagger}(b) \big) ( \targetstate )
\end{equation*}
is a $4 \times 4$ operator. Similarly, let $W(a, b)$ be the $4 \times 4$ Bell operator constructed from local qubit observables corresponding to angles $a$ and $b$ for Alice and Bob, respectively. Thanks to the block structure, computing the lowest eigenvalue of $K - sW$ simplifies to
\begin{equation*}
\mu(s) := \inf \, \lmin ( K - s W ) = \min_{a, b} \, \lmin \big( K(a, b) - s W(a, b) \big),
\end{equation*}
where the minimisation is performed over the square $(a, b) \in [0, \pitwo] \times [0, \pitwo]$. This procedure is precisely the approach used to derive robust self-testing statements in Ref.~\cite{kaniewski16b}. In the following section, we apply it to the case of the tilted CHSH inequality.
\section{Robust self-testing of all entangled two-qubit states}
\label{sec:tilted-chsh}
In 2012 Ac{\'i}n, Pironio and Massar introduced a family of Bell functions which are now commonly referred to as the tilted CHSH family~\cite{acin12a}. The corresponding Bell operator reads
\begin{equation}
\label{eq:tilted-operator}
W_{\alpha} := \alpha A_{0} \otimes \unit + A_{0} \otimes ( B_{0} + B_{1} ) + A_{1} \otimes ( B_{0} - B_{1} ),
\end{equation}
where $\alpharan$ is a parameter. The classical and quantum values of this Bell function equal $\beta_{C} = 2 + \alpha$ and $\beta_{Q} = \sqrt{ 8 + 2 \alpha^{2} }$, respectively. Clearly, for all values of $\alpha$ we have $\beta_{Q} > \beta_{C}$, although the gap vanishes as $\alpha \to 2$. The quantum value can be achieved using a pure state of two qubits $\targetstate^{\alpha} = \ketbraq{\Phi^{\alpha}}_{A'B'}$ for
\begin{equation*}
\ket{\Phi^{\alpha}}_{A'B'} := \cos \theta_{\alpha} \ket{ u_{0} }_{A'} \ket{ v_{0} }_{B'} + \sin \theta_{\alpha} \ket{ u_{1} }_{A'} \ket{ v_{1} }_{B'},
\end{equation*}
where $\{ \ket{ u_{0} }, \ket{ u_{1} } \}$, $\{ \ket{ v_{0} }, \ket{ v_{1} } \}$ are some orthonormal bases on a qubit and
\begin{equation}
\label{eq:alpha-theta}
\theta_{\alpha} := \frac{1}{2} \arcsin \bigg( \sqrt{ \frac{4 - \alpha^{2}}{4 + \alpha^{2}} } \bigg).
\end{equation}
While the optimal observables of Alice are always maximally incompatible, which corresponds to setting $a = \pifour$ in Eqs.~\eqref{eq:A0-def} and~\eqref{eq:A1-def}, the optimal angle on Bob's side changes with $\alpha$ according to
\begin{equation*}
\balphastar.
\end{equation*}
Performing these measurements on this particular state turns out to be essentially the only manner of achieving the maximal violation, i.e.~this Bell inequality is a self-test~\cite{yang13a, bamps15a}. Since the range $\alpharan$ is mapped onto $\theta_{\alpha} \in (0, \pifour]$, it allows us to self-test every pure entangled state of two qubits.

Clearly, setting $\alpha = 0$ yields the CHSH inequality for which the {\meth} method gives strong self-testing bounds~\cite{kaniewski16b} and in this work we apply this approach to the entire range $\alpharan$.

Before stating the conjectured bound, let us briefly explain the construction of extraction channels and the choice of constants $s_{\alpha}$ and $\mu_{\alpha}$. The optimal channels for the CHSH case correspond to full dephasing in $\X$ for $a = 0$, full dephasing in $\Z$ for $a = \pitwo$ and identity channel for $a = \pifour$. This choice is correct for Alice, because her optimal angle is always $\pifour$, but for Bob we must introduce a modification which shifts the occurence of the identity channel to his optimal angle $b_{\alpha}^{*}$. This modification can be achieved by defining an \emph{effective angle} which uniformly extends the interval $[0, b_{\alpha}^{*}]$ to $[0, \pifour]$ and simultaneously shrinks the interval $[b_{\alpha}^{*}, \pitwo]$ to $[\pifour, \pitwo]$. After this modification one can check that this choice of channels performs well on the vertices of the square $(a, b) \in \{ (0, 0), (0, \pitwo), (\pitwo, 0), (\pitwo, \pitwo) \}$ and the point of maximal violation $(a, b) = ( \pifour, b_{\alpha}^{*} )$. We choose the constant $s_{\alpha}$ so that the smallest eigenvalue of the operator $K - s_{\alpha} W$ occurs at multiple points $(a,b)$. In the case of CHSH, i.e.~for $\alpha = 0$, we can obtain the same smallest eigenvalue on all the vertices and the point of maximal violation. However, the case of $\alpha > 0$ is less symmetric and the optimal choice of $s_{\alpha}$ only equalises the smallest eigenvalue at two vertices and the point of maximal violation. Since the operators corresponding to the 5 special points (the vertices and the point of maximal violation) are easy to analyse (the operators $K$ and $W$ are diagonal in the same basis), our choice of $s_{\alpha}$ and $\mu_{\alpha}$ is given by analytic expressions. One can then check that the resulting operator inequality holds at these points for the entire range of $\alpharan$. Unfortunately, verifying the operator inequality on the rest of the square turns out to be much harder and we were not able to do it analytically. However, since the parameter space is bounded ($\alpharan$, $a, b \in [0, \pitwo]$), one can generate a grid over this space and check the operator inequality at those points numerically. We have found that the operator inequality holds up to numerical accuracy (see Appendix~\ref{app:robust-two-qubits} for details), which lends support to the following conjecture.
\begin{conj}
\label{conj:tilted-chsh}
Let $\alpharan$ and let $\inputstate$ be a bipartite quantum state which achieves the tilted CHSH violation of $\beta_{\alpha} := \tr ( W_{\alpha} \inputstate )$, where $W_{\alpha}$ is the Bell operator defined in Eq.~\eqref{eq:tilted-operator}. Then, the extractability of $\targetstate^{\alpha}$ from $\inputstate$ satisfies
\begin{equation*}
\Xi( \inputstate \rightarrow \targetstate^{\alpha} ) \geq s_{\alpha} \cdot \beta_{\alpha} + \mu_{\alpha}
\end{equation*}
for
\params
\end{conj}
In Fig.~\ref{fig:tilted-chsh} we compare the conjectured bounds with the results obtained by Bancal et al.~using the swap method~\cite{bancal15a}.\footnote{The formulation used in the swap method involves isometries rather than channels, but the auxiliary registers are traced out before computing fidelity with the target state (see Eqs.~(10) and (11) in Ref.~\cite{bancal15a}). Therefore, in both cases we obtain lower bounds on precisely the same quantity (see Appendix~\ref{app:formulation} for more details).}

Note that if we trust the numerical package used to verify the operator inequality, this conjecture could be made into a rigorous bound by explicitly calculating the error term. The error term would consist of two components: the error observed numerically on the grid (for our grid this value is of the order of $10^{-9}$) and the discretisation error. Unfortunately, since both $s_{\alpha}$ and $\mu_{\alpha}$ diverge as $\alpha \to 2$, the discretisation error would necessarily diverge in this limit. Therefore, no finite grid enables us to obtain certified bounds for $\alpha$ arbitrarily close to $2$.
\begin{figure}[h]
\includegraphics[width=1.00\columnwidth]{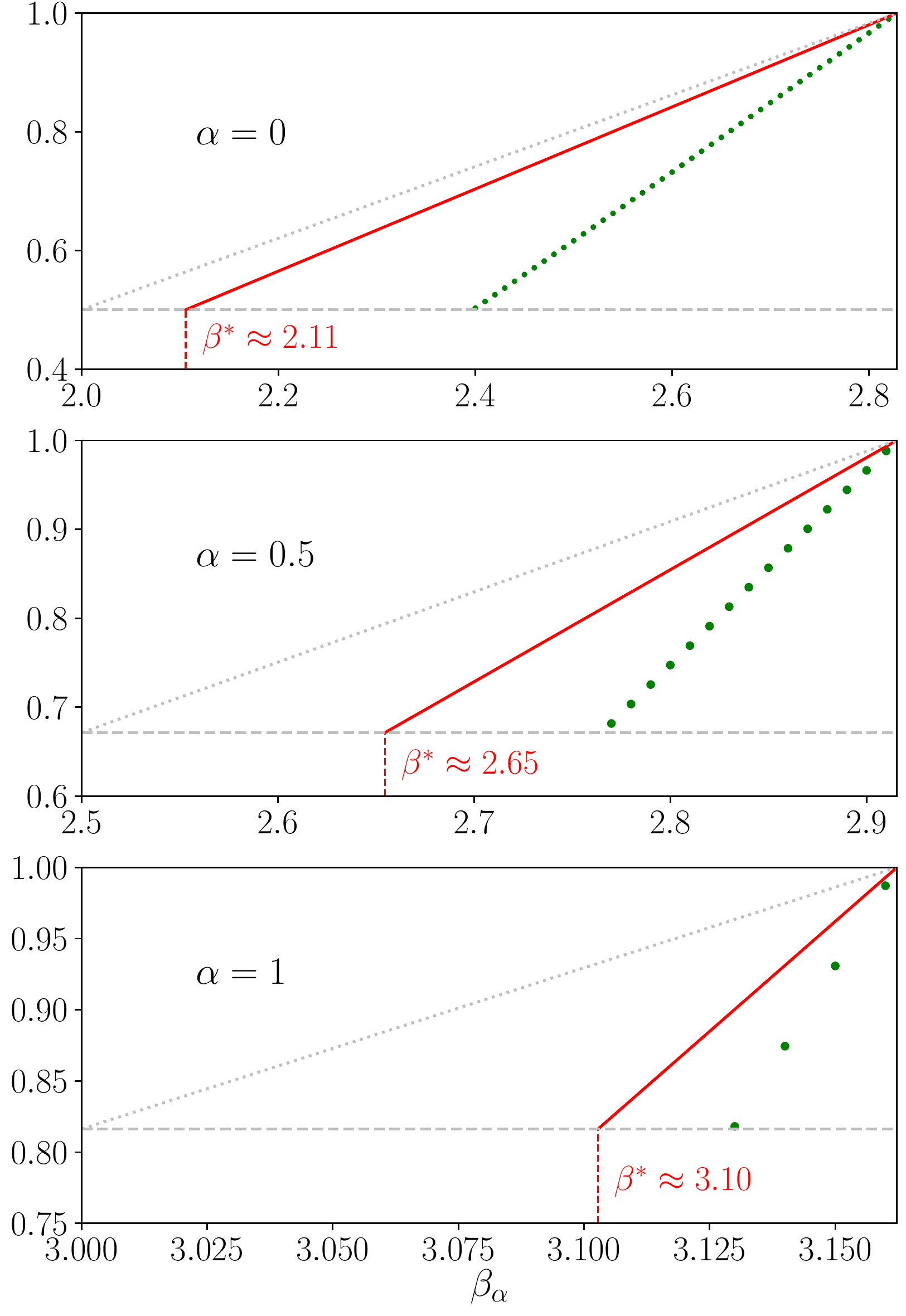}
\caption{Comparison of the conjectured lower bounds (solid line) with the previous results of Bancal et al.~(green points) \cite{bancal15a}. The range of $\beta_{\alpha}$ is chosen to cover the entire range between the classical and the quantum values. The dashed horizontal line indicates the trivial lower bound, whereas the dotted line corresponds to the upper bound given in Eq.~\eqref{eq:upper-bound}. Vertical dashed lines mark the threshold violation $\beta^{*}$ above which the statement becomes non-trivial. The case of $\alpha = 0$ corresponds to the self-testing bound for the CHSH inequality derived in Ref.~\cite{kaniewski16b}.}
\label{fig:tilted-chsh}
\end{figure}
\section{Nontrivial threshold violation for the CHSH inequality}
\label{sec:nontrivial-threshold}
In Ref.~\cite{kaniewski16b} the {\meth} method was used to derive robust bounds on self-testing the singlet\footnote{As explained in Section~\ref{sec:formulation} in the context of self-testing it is only the Schmidt coefficients that matter, so we use the term singlet to mean any (fixed) maximally entangled state of two qubits.} using the CHSH inequality. The resulting statement is nontrivial for any violation exceeding the threshold value of $\betaCHSH := ( 16 + 14 \sqrt{2} )/17 \approx 2.11$ (recall that for the CHSH inequality we have $\beta_{C} = 2$ and $\beta_{Q} = 2 \sqrt{2}$). We have tried to improve on this result, but we have not succeeded. In fact, the dephasing channels specified in the original paper seem to be by far the best choice.

This phenomenon made us wonder whether the existence of a threshold is an inherent feature of quantum mechanics, independent of the proof technique. In other words, maybe one can only make a self-testing statement for sufficiently large violations? The example below shows that this is indeed the case. More specifically, we have constructed a bipartite state which violates the CHSH inequality, but whose singlet extractability does not exceed the separable threshold of $\frac{1}{2}$. In this section we explain the construction of the state, briefly outline the idea of the proof and discuss the implications of this result, while the technical details can be found in Appendix~\ref{app:counterexample}.

Suppose that the system of Alice (Bob) consists of two subsystems: a three-dimensional classical register denoted by $X$ ($Y$) and a qubit denoted by $A$ ($B$). Consider the joint state
\begin{equation*}
\defcq{\rhoXYAB}{p_{xy}}{\rho_{AB}},
\end{equation*}
where $\{p_{xy}\}$ is a normalised probability distribution over $x, y \in \{0, 1, 2\}$ and $\inputstate^{xy}$ are some normalised two-qubit states to be specified later. The observables of Alice are given by
\begin{equation}
\label{eq:observables}
\begin{aligned}
A_{0} &= \ketbraq{0}_{X} \otimes \Z_{A} + \ketbraq{1}_{X} \otimes \Z_{A} + \ketbraq{2}_{X} \otimes \Z_{A},\\
A_{1} &= \ketbraq{0}_{X} \otimes \Z_{A} + \ketbraq{1}_{X} \otimes \X_{A} + \ketbraq{2}_{X} \otimes (-\Z)_{A}.
\end{aligned}
\end{equation}
The observables of Bob are precisely the same, but act on subsystems $Y$ and $B$ instead of $X$ and $A$. Computing the CHSH operator\footnote{The CHSH operator is obtained by setting $\alpha = 0$ in Eq.~\eqref{eq:tilted-operator}.} gives
\begin{equation*}
\defcq{W}{}{W_{AB}},
\end{equation*}
where $W_{AB}^{xy}$ are the resulting two-qubit operators. Let us arrange the 9 possible combination of $(x, y)$ on a $3 \times 3$ grid, where one axis corresponds to $x$ and the other axis corresponds to $y$. We will refer to the point $(x, y) = (1, 1)$ as ``the centre'', while the remaining 8 points constitute ``the frame''. The centre allows for the optimal CHSH violation and we choose $\inputstate^{11}$ to be the corresponding eigenstate of $W_{AB}^{11}$, i.e.
\begin{equation*}
\tr ( W_{AB}^{11} \inputstate^{11} ) = 2 \sqrt{2}.
\end{equation*}
For all the points on the frame the two-qubit operator is a product operator whose eigenvalues are $\{-2, 2\}$. We choose the states $\inputstate^{xy}$ to be classically-correlated and satisfy
\begin{equation*}
\tr ( W_{AB}^{xy} \inputstate^{xy} ) = 2.
\end{equation*}
Clearly, this setup violates the CHSH inequality as long as $p_{11} > 0$.

Now we would like to show that there exists a probability distribution satisfying $p_{11} > 0$ such that the resulting state $\rhoXYAB$ has singlet extractability of $\frac{1}{2}$. In general this is a hard task, as we must show that this value cannot be exceeded regardless of the choice of the extraction channels. Fortunately, the presence of classical registers significantly simplifies the problem due the following observation: any quantum channel that acts simultaneously on classical and quantum registers can be simulated by first reading off the value of the classical register and then applying a particular quantum channel to the quantum register (for completeness we provide a proof, see Lemma~\ref{lem:classical-quantum-channel} in Appendix~\ref{app:counterexample}). This observation implies that instead of considering channels from $\cL( \amsbb{C}^{3} \otimes \amsbb{C}^{2} )$ to $\cL( \amsbb{C}^{2} )$, it suffices to consider triples (one corresponding to each value of the classical register) of qubit ($\cL( \amsbb{C}^{2} ) \to \cL( \amsbb{C}^{2} )$) channels.

All the states on the frame are classically-correlated, but the local bases are different for different points. In fact, one can show that the only strategy that achieves optimal extraction (i.e.~fidelity of $\frac{1}{2}$) on \emph{all} the frame points corresponds to erasing the initial state and replacing it with a fixed product state. This operation is achieved precisely by the full amplitude-damping channel. On the other hand, in order to preserve entanglement of the state in the centre, one should apply some non-destructive channels, e.g.~unitaries. These two requirements are highly incompatible and this incompatibility is precisely what our proof hinges on. We choose a probability distribution concentrated on the frame, which forces Alice and Bob to perform channels close to full amplitude damping and we show that such channels necessarily destroy the entanglement present in the centre. The proof, which consists of a long sequence of elementary inequalities, can be found in Appendix~\ref{app:counterexample}.
\begin{prop}
\label{prop:threshold}
There exists a bipartite state $\rhoXYAB$ which produces a CHSH violation of $\beta \approx \thrviol$, but nevertheless exhibits a singlet extractability of $\frac{1}{2}$.
\end{prop}
This result can be interpreted in several ways. First of all, it implies that self-testing of the singlet using the CHSH inequality is only possible above some threshold. We find this insight rather surprising, since it shows that self-testing scenarios can be split up into two classes depending on whether they exhibit a threshold (like the CHSH inequality) or not (like the Mermin inequality~\cite{kaniewski16b}). Intuitively, one would conjecture that the presence of a threshold is generic and only in some special circumstances can we make self-testing statements arbitrarily close to the classical value $\beta_{C}$. Note that the Mermin inequality is \emph{frustration-free} in the sense that the optimal quantum realisation simultaneously saturates every term of the Bell operator (contrary to the CHSH inequality). We conjecture that frustration-freeness is the source of strong self-testing properties.

We do not know what the exact threshold for the CHSH inequality is, but it must lie in between $\thrviol$ and $\betaCHSH \approx 2.11$. The analysis we perform could certainly be tightened to improve the lower limit of this interval, but one cannot hope for a significant improvement using our method.

It is important to realise that our result crucially relies on choosing the extractability as the quantity relevant for the task of self-testing and one can ask whether the same threshold phenomenon appears if we replace the fidelity with some other distance measure such as the trace distance. While we do not have a definite answer to this question, we would like to point out that extractability is \emph{the only quantity} for which robust self-testing statements have been proven, i.e.~it seems to be the most ``forgiving'' one. We therefore conjecture that if a threshold occurs for the extractability, it will also appear for any other quantity that accurately captures the task of self-testing (although the actual threshold values will, of course, be different).

We have shown that from the extractability point of view the state $\rhoXYAB$ is as uninteresting as any separable state, but it is clear that the entanglement becomes accessible when more general transformations are allowed. If we allow for non-deterministic entanglement extraction (Alice and Bob apply a local extraction map which either succeeds or fails and we only care about the performance if they both succeed), all the entanglement can be extracted. In a similar fashion the entanglement becomes accessible if we allow classical communication between Alice and Bob, i.e.~we perform entanglement distillation. One could therefore ask whether a stronger counterexample could be found, in which we find a state which is not only non-extractable but also non-distillable. Such a counterexample is, however, not possible, because every state that violates the CHSH inequality is necessarily distillable~\cite{masanes06a}.

At first glance our result seems related to the celebrated conjecture of Peres stating that undistillable states do not violate Bell inequalities~\cite{peres99a} (recently disproved by V{\'e}rtesi and Brunner~\cite{vertesi14a}), but this similarity is rather superficial. Distillability is a fundamental property of entanglement and does not require any particular reference state. Singlet extractability, on the other hand, is defined with respect to a specific target state and is tailored specifically to the task of self-testing.
\section{Conclusions and open questions}
\label{sec:conclusions}
In this work we have focused on the problem of self-testing in the channel formulation as proposed by Bardyn et al.~\cite{bardyn09a}. We have discussed the recently proposed {\meth} method and applied it to the tilted CHSH inequality. Moreover, we have shown that self-testing using the CHSH inequality is only possible above some threshold, which implies the existence of two fundamentally different classes of self-testing scenarios.

Let us conclude by presenting a couple of directions for future research. The first natural extension would be to look at scenarios with more than two parties, but still only \changes{two inputs and outputs} per party. The family of Mermin-Ardehali-Belinskii-Klyshko inequalities~\cite{mermin90a, ardehali92a, belinskii93a} is a promising candidate because it is permutation-symmetric and the optimal observables are precisely the same as for the CHSH and Mermin inequalities. We therefore expect that applying the same channels could already give satisfactory results. A more challenging goal is to apply the {\meth} method to scenarios going beyond Jordan's lemma, i.e.~where the number of inputs or outputs is higher than $2$. As this is not an easy task, it might be more tractable in a more restrictive setup, e.g.~in a semi-device-independent scenario where one of the parties is trusted (equivalent to steering~\cite{gheorghiu17a, supic16b, goswami18a}). The {\meth} method has been successfully applied to prepare-and-measure scenarios in which the transmitted system is a qubit~\cite{tavakoli18a} and one might also try to apply it to higher-dimensional cases (although one should remember that they are self-tests in a weaker sense~\cite{farkas18a}).

Another important concept that arises from this work is the threshold violation. We have shown that the CHSH inequality exhibits a threshold violation, but we have not pinned down the number. Computing the exact number is likely to be hard and, moreover, the actual value might depend on the specific formulation of self-testing, which makes it less interesting from a fundamental point of view. However, we see the sheer existence of a threshold as something that deserves a better understanding. We would first like to know whether there exists an alternative natural formulation of the self-testing problem for which the threshold does not appear. If that is not the case, it would be interesting to find out which features of the Bell inequality determine whether it exhibits a threshold or not and which of the two behaviours is generic. We would also like to have an example of a bipartite inequality without a threshold.

Let us finish by pointing out that while the current formulation of self-testing works well in some scenarios, there is some recent evidence that the problem of deducing properties of quantum systems from statistics alone is generically much harder, particularly in multipartite scenarios~\cite{goh18a}. This evidence motivates more relaxed formulations of the problem, where instead of pinning down the exact state, we are happy to obtain a lower bound on some entanglement measure~\cite{bancal11a, pal11a, moroder13a, arnonfriedman17a, arnonfriedman19a}.
\section*{Acknowledgements}
We thank Jean-Daniel Bancal for sharing numerical data with us. TC would like to thank the QMATH Centre and the Erasmus+ Programme for the financial support of his research stay. JK acknowledges support from the National Science Centre, Poland (grant no.~2016/23/P/ST2/02122). This project is carried out under POLONEZ programme which has received funding from the European Union's Horizon 2020 research and innovation programme under the Marie Sk{\l}odowska-Curie grant agreement no.~665778. CS is supported by a NWO VIDI grant (Project No.~639.022.519).
\bibliography{library}
\appendix
\onecolumngrid
\section{Formulations of the self-testing problem}
\label{app:formulation}
In this appendix we discuss possible formulations of the self-testing problem. In the first part we show that the three commonly used formulations are equivalent. In the second part we explain how to make these formulations robust and discuss the relations between the resulting inequivalent measures for robust self-testing.
\subsection{Exact self-testing definitions}
A linear map $V : \sH_{A} \to \sH_{B}$ is called an isometry if it satisfies $V^{\dagger} V = \unit_{A}$. For a Hilbert space $\sH$ let $\cS(\sH)$ be the set of density operators acting on $\sH$. Let $\sH_{X}$ and $\sH_{X'}$ for $X \in \{ A, B \}$ be finite-dimensional Hilbert spaces. The target state $\targetstate \in \cS( \sH_{A'} \otimes \sH_{B'} )$ is pure ($\targetstate^{2} = \targetstate$) and its marginals ($\Phi_{A'} := \tr_{B'} \Phi_{A'B'}$ and $\Phi_{B'} := \tr_{A'} \Phi_{A'B'}$) are full-rank ($\rk(\Phi_{X'}) = \dim( \sH_{X'} )$ for $X \in \{ A, B \}$, which immediately implies $\dim( \sH_{A'} ) = \dim( \sH_{B'} )$). The input state $\inputstate \in \cS( \sH_{A} \otimes \sH_{B} )$ is arbitrary.
\begin{prop}
The following three statements are equivalent.
\begin{enumerate}
\item[(1)] There exist completely positive trace-preserving maps $\Lambda_{X} : \cL( \sH_{X} ) \to \cL( \sH_{X'} )$ such that
\begin{equation}
\label{eq:statement-1}
( \Lambda_{A} \otimes \Lambda_{B} ) ( \inputstate ) = \targetstate.
\end{equation}
\item[(2)] There exist Hilbert spaces $\sH_{X''}$, isometries $V_{X} : \sH_{X} \to \sH_{X'} \otimes \sH_{X''}$ and an auxiliary state $\auxstate \in \cS( \sH_{A''} \otimes \sH_{B''} )$ such that
\begin{equation}
\label{eq:statement-2}
V \inputstate V^{\dagger} = \targetstate \otimes \auxstate,
\end{equation}
where $V = V_{A} \otimes V_{B}$ is the combined isometry.
\item[(3)] There exist Hilbert spaces $\sH_{X'''}$, isometries $W_{X} : \sH_{X'} \otimes \sH_{X'''} \to \sH_{X}$ and an auxiliary state $\tau_{A'''B'''} \in \cS( \sH_{A'''} \otimes \sH_{B'''} )$ such that
\begin{equation}
\label{eq:statement-3}
\inputstate = W ( \targetstate \otimes \tau_{A'''B'''} ) W^{\dagger},
\end{equation}
where $W = W_{A} \otimes W_{B}$ is the combined isometry.
\end{enumerate}
\end{prop}
Before proceeding to the proof, let us sketch how the three formulations are connected. The equivalence between $(1)$ and $(2)$ is a direct consequence of Naimark's dilation theorem. The relation between $(2)$ and $(3)$, on the other hand, is more subtle and deserves a brief discussion. If the isometry $V$ in Eq.~\eqref{eq:statement-2} happens to be a unitary, we can just move it to the other side to obtain Eq.~\eqref{eq:statement-3} and the equivalence is trivial. However, if the dimensions do not match, i.e.~when $\dim(\sH_{X})$ is not a multiple of $\dim(\sH_{X'})$, the isometry $V_{X}$ cannot be a unitary and cannot be inverted. Then, the solution is to invert it only on the support of the state $\Phi_{X'} \otimes \sigma_{X''}$ and the construction proving $(2) \to (3)$ does precisely that. The proof of $(3) \to (2)$ proceeds analogously.
\begin{proof}
To see $(1) \to (2)$ we construct Naimark's dilation of the extraction channels. This gives us Hilbert spaces $\sH_{A''}, \sH_{B''}$ and local isometries $V_{A}, V_{B}$ such that
\begin{equation*}
V \inputstate V^{\dagger} = \eta_{A'B'A''B''}
\end{equation*}
and $\tr_{A''B''} \eta_{A'B'A''B''} = \targetstate$. Since the reduced state on $A'B'$ is pure, it must be uncorrelated from the state on $A''B''$, which concludes the proof. The opposite direction is easy: the extraction channel corresponds to applying the isometry and tracing out the auxiliary system.

To prove $(2) \to (3)$ we explicitly construct a new Hilbert space, isometries and an auxiliary state. Let us start by showing a simple implication of Eq.~\eqref{eq:statement-2}. Tracing out one of the systems gives
\begin{equation*}
V_{X} \rho_{X} V_{X}^{\dagger} = \Phi_{X'} \otimes \sigma_{X''}.
\end{equation*}
If two operators are equal, their supports must be equal too. Moreover, the support of a tensor product is the tensor product of the supports. Let $\Pi_{X}$ and $\Pi_{X''}$ be the projectors on the supports of $\rho_{X}$ and $\sigma_{X''}$, respectively. Since $\Phi_{X'}$ is full-rank, we obtain
\begin{equation}
\label{eq:support-equality}
V_{X} \Pi_{X} V_{X}^{\dagger} = \unit_{X'} \otimes \Pi_{X''}.
\end{equation}
We can now proceed to the construction. Consider a Hilbert space $\sH_{X'''}$ such that $\dim ( \sH_{X'''} ) = \tr( \Pi_{X''} )$ equipped with an isometry $T_{X} : \sH_{X'''} \to \sH_{X''}$ satisfying
\begin{equation}
\label{eq:Tx-image}
T_{X} T_{X}^{\dagger} = \Pi_{X''}.
\end{equation}
Define
\begin{equation}
\label{eq:tau-definition}
\tau_{A'''B'''} := ( T_{A}^{\dagger} \otimes T_{B}^{\dagger} ) \auxstate ( T_{A} \otimes T_{B} ).
\end{equation}
To see that $\tau_{A'''B'''}$ is a valid state we need to check that it is positive semidefinite and of unit trace. The first property is clear (if $A \geq 0$, then $X^{\dagger} A X \geq 0$ for any $X$), while for the second property we first observe that
\begin{equation*}
\tr( \tau_{A'''B'''} ) = \tr \big( ( \Pi_{A''} \otimes \Pi_{B''} ) \auxstate \big)
\end{equation*}
and then recall that projecting on the local supports does not affect the state, i.e.
\begin{equation*}
( \Pi_{A''} \otimes \Pi_{B''} ) \auxstate = \auxstate.
\end{equation*}
Define $W_{X} : \sH_{X'} \otimes \sH_{X'''} \to \sH_{X}$ as
\begin{equation*}
W_{X} := \Pi_{X} V_{X}^{\dagger} ( \unit_{X'} \otimes T_{X} ).
\end{equation*}
To see that $W_{X}$ is an isometry compute
\begin{align*}
W_{X}^{\dagger} W_{X} &= ( \unit_{X'} \otimes T_{X}^{\dagger} ) V_{X} \Pi_{X} V_{X}^{\dagger} ( \unit_{X'} \otimes T_{X} ) = ( \unit_{X'} \otimes T_{X}^{\dagger} ) ( \unit_{X'} \otimes \Pi_{X''} ) ( \unit_{X'} \otimes T_{X} )\\
&= \unit_{X'} \otimes ( T_{X}^{\dagger} \Pi_{X''} T_{X} ) = \unit_{X'} \otimes \unit_{X'''},
\end{align*}
where in the first line we have used Eq.~\eqref{eq:support-equality}, while the last step relies on Eq.~\eqref{eq:Tx-image}. Finally, we must verify that Eq.~\eqref{eq:statement-3} holds. Writing out the right-hand side gives
\begin{align*}
W \big( \targetstate \otimes \tau_{A'''B'''} \big) W^{\dagger} &= ( \Pi_{A} \otimes \Pi_{B} ) ( V_{A}^{\dagger} \otimes V_{B}^{\dagger} ) ( \unit_{A'B'} \otimes T_{A} \otimes T_{B} ) ( \targetstate \otimes \tau_{A'''B'''} )\\
& \quad \; ( \unit_{A'B'} \otimes T_{A}^{\dagger} \otimes T_{B}^{\dagger} ) ( V_{A} \otimes V_{B} ) ( \Pi_{A} \otimes \Pi_{B} )\\
&= ( \Pi_{A} \otimes \Pi_{B} ) ( V_{A}^{\dagger} \otimes V_{B}^{\dagger} ) [ \targetstate \otimes ( T_{A} \otimes T_{B} ) \tau_{A'''B'''} ( T_{A}^{\dagger} \otimes T_{B}^{\dagger} ) ] ( V_{A} \otimes V_{B} ) ( \Pi_{A} \otimes \Pi_{B} ).
\end{align*}
We simplify the middle term using Eq.~\eqref{eq:tau-definition}
\begin{equation*}
( T_{A} \otimes T_{B} ) \tau_{A'''B'''} ( T_{A}^{\dagger} \otimes T_{B}^{\dagger} ) = ( \Pi_{A''} \otimes \Pi_{B''} ) \auxstate ( \Pi_{A''} \otimes \Pi_{B''} ) = \auxstate.
\end{equation*}
Therefore,
\begin{align*}
W \big( \targetstate \otimes \tau_{A'''B'''} \big) W^{\dagger} &= ( \Pi_{A} \otimes \Pi_{B} ) ( V_{A}^{\dagger} \otimes V_{B}^{\dagger} ) ( \targetstate \otimes \auxstate ) ( V_{A} \otimes V_{B} ) ( \Pi_{A} \otimes \Pi_{B} )\\
&= ( \Pi_{A} \otimes \Pi_{B} ) \inputstate ( \Pi_{A} \otimes \Pi_{B} ) = \inputstate,
\end{align*}
where the middle step is a direct consequence of Eq.~\eqref{eq:statement-2}.

The proof of $(3) \to (2)$ is, again, a construction. Analogous to the previous argument we find that the projectors on the supports $\Pi_{X}$ and $\Pi_{X'''}$ satisfy
\begin{equation*}
\Pi_{X} = W_{X} ( \unit_{X'} \otimes \Pi_{X'''} ) W_{X}^{\dagger}.
\end{equation*}
Consider a Hilbert space $\sH_{X''}$ (dimension to be specified later) and a linear map $L_{X} : \sH_{X''} \to \sH_{X'''}$ satisfying
\begin{equation*}
L_{X} L_{X}^{\dagger} = \Pi_{X'''}.
\end{equation*}
Let
\begin{equation*}
\auxstate := ( L_{A}^{\dagger} \otimes L_{B}^{\dagger} ) \tau_{A'''B'''} ( L_{A} \otimes L_{B} )
\end{equation*}
and it is easy to check that $\auxstate$ is a valid state. Finally, we need an isometry $R_{X} : \sH_{X} \to \sH_{X'} \otimes \sH_{X''}$ such that the projectors $R_{X} ( \unit_{X} - \Pi_{X} ) R_{X}^{\dagger}$ and $\unit_{X'} \otimes L_{X}^{\dagger} L_{X}$ are orthogonal. Finding such an isometry is possible if the Hilbert space $\sH_{X''}$ is of sufficiently high dimension. A simple dimension counting argument implies that we must choose $\dim( \sH_{X''} )$ to satisfy $\dim( \sH_{X'} ) \cdot \dim( \sH_{X''} ) \geq \dim( \sH_{X} )$.
Define $V_{X} : \sH_{X} \to \sH_{X'} \otimes \sH_{X''}$ as
\begin{equation*}
V_{X} = ( \unit_{X'} \otimes L_{X}^{\dagger} ) W_{X}^{\dagger} + R_{X} ( \unit_{X} - \Pi_{X} ).
\end{equation*}
It is easy to verify that $V_{X}$ is an isometry and that the combined isometry $V := V_{A} \otimes V_{B}$ satisfies Eq.~\eqref{eq:statement-2}.
\end{proof}
\subsection{Robust self-testing measures}
\label{app:robust-measures}
The conditions discussed in the previous section capture the idea that a perfect copy of the target state can be extracted from the real state. If we want to use these quantities in any real-world situation, we need to introduce their approximate versions. In the ideal case we require the existence of some objects (e.g.~channels or isometries) which render equalities~\eqref{eq:statement-1}, \eqref{eq:statement-2}, \eqref{eq:statement-3} true. In the approximate case we will quantify approximate satisfaction of these equalities by computing the fidelity between the left- and right-hand sides and we will maximise this value over all valid objects. Note that instead of using the fidelity, we could use the trace norm as a measure of distance, but since we are not aware of any robust results involving the trace distance, we do not discuss it here.

The approximate satisfaction of condition~\eqref{eq:statement-1} is quantified by the extractability defined as
\begin{equation*}
\extractability,
\end{equation*}
where the maximisation is taken over all quantum channels from $A$ to $A'$ and $B$ to $B'$, respectively. Basic properties of extractability are discussed in Section~\ref{sec:formulation} of the main text.

Condition~\eqref{eq:statement-2} gives rise to a measure which we call \emph{isometric fidelity} defined as
\begin{equation}
\label{eq:isometric-fidelity-definition}
\Fiso ( \inputstate \to \targetstate ) := \sup_{\auxstate} \sup_{V} F \big( V\inputstate V^{\dagger} , \targetstate \otimes \auxstate \big),
\end{equation}
where the supremum is taken over product isometries $V = V_{A} \otimes V_{B}$, where $V_{X} : \sH_{X} \to \sH_{X'} \otimes \sH_{X''}$, and auxiliary states $\auxstate \in \cS( \sH_{A''} \otimes \sH_{B''} )$. Perhaps surprisingly, this quantity turns out to be equal to the extractability as long as the target state is pure \cite{ribeirodahlberg18}.
\begin{prop}
\label{prop:extractability-equals-isometric}
Let $\inputstate$ be an arbitrary input state and $\targetstate$ be an arbitrary pure target state. Then,
\begin{equation*}
\Xi( \inputstate \rightarrow \targetstate ) = \Fiso ( \inputstate \to \targetstate ).
\end{equation*}
\end{prop}
\begin{proof}
To see that the extractability is never smaller than the isometric fidelity it suffices to realise that every local isometry can be turned into an extraction channel by performing a partial trace. Since the fidelity is non-decreasing under tracing out, we immediately conclude that
\begin{equation*}
\Xi( \inputstate \to \targetstate ) \geq \Fiso ( \inputstate \to \targetstate ).
\end{equation*}
To show that this inequality holds as an equality we use Uhlmann's theorem. Let $\Lambda_{A}$ and $\Lambda_{B}$ be a pair of extraction channels that achieves optimal fidelity in the definition of extractability, i.e.~if
\begin{equation*}
\zeta_{A'B'} := ( \Lambda_{A} \otimes \Lambda_{B} ) ( \inputstate ),
\end{equation*}
then
\begin{equation*}
\Xi( \inputstate \rightarrow \targetstate ) = F( \zeta_{A'B'} , \targetstate).
\end{equation*}
Uhlmann's theorem implies that the fidelity between two mixed states equals the highest achievable fidelity between their purifications and, moreover, that one of the purifications can be fixed. In our case we pick a specific purification of $\zeta_{A'B'}$. Let $\rho_{ABE}$ be a purification of $\rho_{AB}$, for $X \in \{ A, B \}$ let $V_{X} : \sH_{X} \to \sH_{X'} \otimes \sH_{X''}$ be Naimark's dilation of the extraction channel $\Lambda_{X}$ and finally let $V_{AB} := V_{A} \otimes V_{B}$. Then, the state
\begin{equation*}
\zeta_{A'B'A''B''E} := (V_{AB} \otimes \unit_{E}) \rho_{ABE} (V_{AB}^{\dagger} \otimes \unit_{E})
\end{equation*}
is a purification of $\zeta_{A'B'}$. By Uhlmann's theorem there exists a purification of $\targetstate$, which we denote by $\gamma_{A'B'A''B''E}$, such that
\begin{equation}
\label{eq:fidelity-equality}
F( \zeta_{A'B'} , \targetstate) = F( \zeta_{A'B'A''B''E}, \gamma_{A'B'A''B''E} ).
\end{equation}
However, since $\targetstate$ is already pure, all its purifications are of the form
\begin{equation*}
\gamma_{A'B'A''B''E} = \targetstate \otimes \gamma_{A''B''E}
\end{equation*}
for some pure state $\gamma_{A''B''E}$. Since the fidelity is non-decreasing under tracing out, we have
\begin{equation*}
F( \zeta_{A'B'A''B''E}, \targetstate \otimes \gamma_{A''B''E} ) \leq F( \zeta_{A'B'A''B''}, \targetstate \otimes \gamma_{A''B''} ) \leq F( \zeta_{A'B'} , \targetstate),
\end{equation*}
which together with Eq.~\eqref{eq:fidelity-equality} implies that
\begin{equation*}
F( \zeta_{A'B'A''B''}, \targetstate \otimes \gamma_{A''B''} ) = F( \zeta_{A'B'} , \targetstate).
\end{equation*}
The left-hand side is a lower bound on the isometric fidelity, whereas the right-hand side by construction equals the extractability, which concludes the proof.
\end{proof}

To finish our discussion of the isometric fidelity, let us point out that in the literature one sometimes sees the isometries in Eqs.~\eqref{eq:statement-2} and \eqref{eq:isometric-fidelity-definition} replaced by unitaries, but using unitaries is strictly speaking not correct. For instance the unitary version of isometric fidelity has the unpleasant feature that it is not defined for all input states. The existence of a unitary $U_{A} : \sH_{A} \to \sH_{A'} \otimes \sH_{A''}$ implies that $\dim( \sH_{A} ) = \dim( \sH_{A'} ) \cdot \dim( \sH_{A''} )$. Since the dimension of the auxiliary Hilbert space $\sH_{A''}$ must be an integer, unitarity requires that the dimension of the Hilbert space $\sH_{A}$ is a multiple of the dimension of the target Hilbert space $\sH_{A'}$, which does not have to be the case. Clearly, a measure which is not defined for all states is not suitable for the purpose of making self-testing statements.

Finally, condition~\eqref{eq:statement-3} gives rise to the Mayers-Yao (MY) fidelity defined as
\begin{equation*}
\FMY ( \inputstate \to \targetstate ) := \sup_{\auxstate} \sup_{W} F \big( \inputstate , W ( \targetstate \otimes \auxstate ) W^{\dagger} \big),
\end{equation*}
where the supremum is taken over product isometries $W = W_{A} \otimes W_{B}$ for $W_{X} : \sH_{X'} \otimes \sH_{X''} \to \sH_{X}$ and auxiliary states $\auxstate \in \cS( \sH_{A''} \otimes \sH_{B''} )$. However, this quantity suffers from the same problem: it is not defined for all states, e.g.~when $\dim( \sH_{A} ) < \dim( \sH_{A'} ) $.
\section{Robust self-testing of two-qubit states}
\label{app:robust-two-qubits}
In this appendix we provide the details of the argument discussed in Section~\ref{sec:tilted-chsh}. In the first part we give the definitions of the extraction channels and compute all the operators appearing in the operator inequality. In the second part we discuss the numerical evidence supporting the conjecture.
\subsection{Operator inequality}
Let us start by writing down the Bell operator. Recall that the observables of Alice and Bob are parametrised by
\begin{align*}
A_{r} &:= \cos a \, \X + (-1)^{r} \sin a \, \Z,\\
B_{r} &:= \cos b \, \X + (-1)^{r} \sin b \, \Z
\end{align*}
for $r \in \{0, 1\}$. For these observables the tilted CHSH operator defined in Eq.~\eqref{eq:tilted-operator} reads
\begin{align*}
W_{\alpha}(a, b) &= \alpha ( \cos a \, \X + \sin a \, \Z ) \otimes \unit + 2 \cos a \cos b \, \X \otimes \X + 2 \cos a \sin b \, \X \otimes \Z\\
&\hspace{10pt}+ 2 \sin a \cos b \, \Z \otimes \X - 2 \sin a \sin b \, \Z \otimes \Z.
\end{align*}
The optimal violation is achieved for $a^{*} := \pifour$ and
\begin{equation}
\label{eq:balphastar}
b_{\alpha}^{*} := \arcsin \bigg( \sqrt{ \frac{4 - \alpha^{2}}{8} } \bigg).
\end{equation}
The corresponding optimal state is given by
\begin{equation}
\label{eq:optimal-state}
\Phi_{\alpha} := \frac{1}{4} \bigg( \unit \otimes \unit + \sqrt{ \frac{ 2 \alpha^{2} }{ 4 + \alpha^{2} } } \bigg[ \frac{\X + \Z}{\sqrt{2}} \otimes \unit + \unit \otimes \X \bigg] + \frac{ \X + \Z }{\sqrt{2}} \otimes \X + \sqrt{ \frac{ 4 - \alpha^{2} }{ 4 + \alpha^{2} } } \bigg[ \Y \otimes \Y + \frac{ \X - \Z }{ \sqrt{2} } \otimes \Z \bigg] \bigg).
\end{equation}
To see that this state is unitarily equivalent to $ \cos \theta \ket{00} + \sin \theta \ket{11}$ for $\theta$ specified in Eq.~\eqref{eq:alpha-theta} note that
\begin{equation*}
\sin 2 \theta = \sqrt{ \frac{ 4 - \alpha^{2} }{ 4 + \alpha^{2} } } \; \nbox{and} \cos 2 \theta = \sqrt{ \frac{ 2 \alpha^{2} }{ 4 + \alpha^{2} } }.
\end{equation*}
The extraction channel for Alice is precisely the channel used in Ref.~\cite{kaniewski16b}:
\begin{equation*}
[\Lambda_{A}(x)] (\rho) := \frac{ 1 + g(x) }{2} \, \rho + \frac{ 1 - g(x) }{2} \, \Gamma(x) \rho \Gamma(x),
\end{equation*}
where
\begin{equation*}
\Gamma(x) :=
\begin{cases}
\X &\nbox{if} x \in [0, \pifour],\\
\Z &\nbox{if} x \in (\pifour, \pitwo]
\end{cases}
\end{equation*}
and
\begin{equation*}
\gfundef.
\end{equation*}
It is easy to check that $x = \pifour$ gives the identity channel, whereas $x = 0$ and $x = \pitwo$ correspond to full dephasing. The channel of Bob has the same form except that the identity channel should arise for the angle $b_{\alpha}^{*}$ defined in Eq.~\eqref{eq:balphastar}. Let us define the ``effective angle'' $h_{\alpha}(x)$ as a piecewise linear function which maps the interval $[0, b_{\alpha}^{*}]$ onto $[0, \pifour]$ and $[b_{\alpha}^{*}, \pitwo]$ onto $[\pifour, \pitwo]$:
\begin{equation*}
h_{\alpha}(x) :=
\begin{cases}
\frac{\pi}{4} \cdot \frac{x}{b_{\alpha}^{*}} &\nbox{if} x \in [0, b_{\alpha}^{*}],\\
\frac{\pi}{2} - \frac{\pi}{4} \cdot \frac{ \pi - 2x }{ \pi - 2 b_{\alpha}^{*} } &\nbox{if} x \in (b_{\alpha}^{*}, \pitwo].
\end{cases}
\end{equation*}
These definitions allow us to write the extraction channel of Bob as
\begin{equation*}
\Lambda_{B}(x) := \Lambda_{A}( h_{\alpha}(x) ).
\end{equation*}
The operator $K_{\alpha}(a, b)$ is obtained by applying the dual channels to the ideal state given in Eq.~\eqref{eq:optimal-state}. Since the dephasing channels are self-dual, we have
\begin{equation*}
K_{\alpha}(a, b) := \big( \Lambda_{A}(a) \otimes \Lambda_{B}(b) \big) ( \Phi_{\alpha} ).
\end{equation*}
The operator inequality~\eqref{eq:operator-inequality} is equivalent to the operator
\begin{equation*}
T_{\alpha}(a, b) := K_{\alpha}(a, b) - s_{\alpha} W_{\alpha}(a, b) - \mu_{\alpha} \unit
\end{equation*}
being positive semidefinite for
\params
Since the dephasing basis changes at $a = \pifour$ and $b = b_{\alpha}^{*}$, there are in principle four distinct cases that need to be considered. In the case of CHSH the presence of symmetries allows one to reduce the analysis of the entire square ($[0, \pitwo] \times [0, \pitwo]$) to a single quarter ($[0, \pifour] \times [0, \pifour]$). In the tilted case this symmetry is partially broken, but we still have
\begin{equation}
\label{eq:unitary-equivalence}
T_{\alpha}(a, b) = U T_{\alpha}( \pitwo - a, \, b ) U^{\dagger},
\end{equation}
where
\begin{equation}
\label{eq:U-def}
U := \frac{ X + Z }{\q} \otimes X.
\end{equation}
This observation implies that it suffices to analyse the half of the square corresponding to $a \in [0, \pifour]$.
\subsection{Numerical evidence}
Our goal is to gather evidence that the operator $T_{\alpha}(a, b)$ is positive semidefinite for $\alpharan, a \in [0, \pifour], b \in [0, \pitwo]$. For this purpose, we have generated a grid over the parameter space in the following manner.
\begin{itemize}
\item We have chosen $\alpha$ in the range $[0, 1.999]$ with a step size of $0.001$.
\item We have discretised the angle of Alice by splitting the interval $[0, \pi/4]$ into $99$ equally-spaced intervals $[a_k, a_{k+1}]$, where $a_1 = 0$, $a_{100} = \pi/4$ and $1\leq k \leq 100$. Similarly, for the angle of Bob we have discretised $[0, \pi/2]$ as intervals $[b_m, b_{m+1}]$ of equal length, with $b_1 = 0$, $b_{200} = \pi/2$ and $0\leq m \leq 200$. For fixed $\alpha$, we thus obtain the grid $\{(a_k, b_m) \mid 1\leq k \leq 100, 1\leq m \leq 200\}$.
\end{itemize}
Using the \texttt{linalg} library from \texttt{Numpy} (a scientific computing package for Python) we have computed the eigenvalues of $T_{\alpha}(a, b)$ at every point of the grid. We have found that the smallest value equals $-1.317 \cdot 10^{-9}$ and occurs for $\alpha = 1.998$. Our code can be freely accessed online \cite{tiltedcode}.
\section{CHSH violation does not imply nontrivial extractability}
\label{app:counterexample}
In this appendix we construct a state which violates the CHSH inequality, but whose singlet extractability does not exceed the trivial value of $\frac{1}{2}$. The proof hinges on two technical propositions and since proving them within the main argument would be rather distracting, let us use them without proofs. Complete proofs can be found in Section~\ref{app:proof-details}.
\subsection{The argument}
Consider a state $\rhoXYAB$ acting on $\sH_{X} \otimes \sH_{Y} \otimes \sH_{A} \otimes \sH_{B}$ for $\sH_{X}, \sH_{Y} \equiv \amsbb{C}^{3}$ and $\sH_{A}, \sH_{B} \equiv \amsbb{C}^{2}$, where subsystems $X$ and $A$ belong to Alice and subsystems $Y$ and $B$ belong to Bob. The state is defined with respect to the CHSH operator corresponding to the observables given in Eq.~\eqref{eq:observables} which reads
\begin{equation*}
\defcq{W}{}{W_{AB}}
\end{equation*}
for the two-qubit operators $W_{AB}^{xy}$ given by
\begin{equation}
\begin{array}{c|ccc}
x \backslash y & 0 & 1 & 2 \\ \hline
0 & 2 \Z \otimes \Z & 2 \Z \otimes \Z & 2 \Z \otimes \Z \\
1 & 2 \Z \otimes \Z & \qquad \X \otimes ( - \X + \Z ) + \Z \otimes ( \X + \Z ) \qquad & 2 \X \otimes \Z \\
2 & 2 \Z \otimes \Z & 2 \Z \otimes \X & - 2 \Z \otimes \Z
\end{array}.
\end{equation}
We choose the state $\rhoXYAB$ to be of the form
\begin{equation*}
\defcq{\rhoXYAB}{p_{xy}}{\rho_{AB}}
\end{equation*}
for some probability distribution $\{ p_{xy} \}_{x, y = 0}^{2}$ and two-qubit states $\rho_{AB}^{xy}$ chosen to satisfy
\begin{equation}
\label{eq:rhoxy}
\ave{ W_{AB}^{xy}, \rho_{AB}^{xy} } =
\begin{cases}
2 \sqrt{2} &\nbox{if} $x = y = 1$,\\
2 & \nbox{otherwise.}
\end{cases}
\end{equation}
The precise form of the states $\rho_{AB}^{xy}$ will be specified later. Recall that we refer to the point $x = y = 1$ as ``the centre'' and the remaining 8 points as ``the frame''. A simple calculation shows that
\begin{equation}
\label{eq:chsh-violation}
\beta = \ave{ W, \rhoXYAB } = 2 + ( 2 \sqrt{2} - 2 ) p_{11},
\end{equation}
i.e.~the CHSH inequality is violated as long as $p_{11} > 0$. Our goal is to prove that there exists a probability distribution satisfying $p_{11} > 0$ and two-qubit states satisfying Eq.~\eqref{eq:rhoxy} such that the resulting state $\rho_{XYAB}$ satisfies
\begin{equation*}
\Xi( \rhoXYAB \rightarrow \targetstate^{+} ) = \frac{1}{2},
\end{equation*}
where $\targetstate^{+}$ is a maximally entangled state of two qubits. The quantity does not depend on which maximally entangled state we choose and for this proof, it is convenient to assume that $\ket{\Phi^{+}} = ( \ket{00} + \ket{11} ) / \sqrt{2}$. By definition of extractability showing existence of a suitable probability distribution and two-qubit states is equivalent to showing that for all local extraction channels $\Lambda_A, \Lambda_B : \cL( \amsbb{C}^{3} \otimes \amsbb{C}^{2} ) \to \cL( \amsbb{C}^{2} )$ we have
\begin{equation*}
F \big( ( \Lambda_{A} \otimes \Lambda_{B} ) ( \rhoXYAB ) , \targetstate^{+} \big) \leq \frac{1}{2}.
\end{equation*}
Since the registers $X$ and $Y$ are classical, instead of optimising over the most general channels from $\cL( \amsbb{C}^{3} \otimes \amsbb{C}^{2} )$ to $\cL( \amsbb{C}^{2} )$ it suffices to consider channels which first read the classical register and then apply a suitable qubit ($\cL( \amsbb{C}^{2} ) \to \cL( \amsbb{C}^{2} )$) channel (see Lemma~\ref{lem:classical-quantum-channel} for details). Let $\Lambda_{A}^{x}$ be the qubit channel of Alice corresponding to the value of the classical register $X$ being $x$ and similarly let $\Lambda_{B}^{y}$ be the qubit channel of Bob corresponding to $Y$ having value $y$. Since the target state is pure, the fidelity equals the inner product which implies
\begin{equation}
\label{eq:fidelity-sum}
F \big( ( \Lambda_{A} \otimes \Lambda_{B} ) ( \rhoXYAB ) , \targetstate^{+} \big) = \ave{ ( \Lambda_{A} \otimes \Lambda_{B} ) ( \rhoXYAB ), \targetstate^{+} } = \sum_{xy} p_{xy} \ave[\big]{ ( \Lambda_{A}^{x} \otimes \Lambda_{B}^{y} )( \rho_{AB}^{xy} ), \targetstate^{+} }.
\end{equation}
The intuition behind the proof goes as follows: there are no extraction channels which perform well both on the frame and in the centre. We make this intuition rigorous in two steps. The following proposition shows that if Alice and Bob perform well on the frame, then the channels $\Lambda_{A}^{1}$ and $\Lambda_{B}^{1}$ must significantly contract the Bloch sphere. Note that in the argument below only six points of the frame are used (we leave the remaining two points undefined).
\begin{prop}
\label{prop:ewav-eigenvalue}
\propewaveigenvalue
\end{prop}
The fact that the channels $\Lambda_{A}^{1}$ and $\Lambda_{B}^{1}$ map the centre of the Bloch sphere to a point close to the boundary means that the input states are to a large extent erased. It is therefore not surprising that applying such channels to a maximally entangled state annihilates most of its entanglement.
\begin{prop}
\label{prop:centre-bound}
\propcentrebound
\end{prop}
These two propositions immediately imply the main result.
\begin{prop}
Let
\begin{equation*}
\defcq{\rhoXYAB}{p_{xy}}{\rho_{AB}},
\end{equation*}
where the states $\rho_{AB}^{xy}$ corresponding to the frame are specified in Proposition~\ref{prop:ewav-eigenvalue}, the state $\rho_{AB}^{11} = \Psi$ is some pure maximally entangled state and the probability distribution is given by
\begin{align*}
p_{00} &= \frac{4}{31} (1 - v), & p_{01} &= p_{10} = \frac{3}{62} (1 - v), & p_{02} &= p_{20} = p_{22} = \frac{8}{31} (1 - v),\\
p_{11} &= v, & p_{12} &= p_{21} = 0
\end{align*}
for $v = 1/597$. This state satisfies
\begin{equation*}
\Xi( \rhoXYAB \to \targetstate^{+} ) = \frac{1}{2}.
\end{equation*}
\end{prop}
\begin{proof}
From Eq.~\eqref{eq:fidelity-sum} we have
\begin{align*}
F \big( ( &\Lambda_{A} \otimes \Lambda_{B} ) ( \rhoXYAB ) , \targetstate^{+} \big) = \sum_{xy} p_{xy} \ave[\big]{ ( \Lambda_{A}^{x} \otimes \Lambda_{B}^{y} )( \rho_{AB}^{xy} ), \targetstate^{+} }\\
&= \sum_{ (x, y) \neq (1, 1) } p_{xy} \ave[\big]{ ( \Lambda_{A}^{x} \otimes \Lambda_{B}^{y} )( \rho_{AB}^{xy} ), \targetstate^{+} } + p_{11} \ave[\big]{ ( \Lambda_{A}^{1} \otimes \Lambda_{B}^{1} )( \Psi ), \targetstate^{+} }.
\end{align*}
The inner product in the first term can be written in terms of $\varepsilon_{xy}$ defined in Proposition~\ref{prop:ewav-eigenvalue}. A direct calculation gives
\begin{equation*}
\sum_{ (x, y) \neq (1, 1) } p_{xy} \ave[\big]{ ( \Lambda_{A}^{x} \otimes \Lambda_{B}^{y} )( \rho_{AB}^{xy} ), \targetstate^{+} } = \sum_{ (x, y) \neq (1, 1) } p_{xy} \Big( \frac{1}{2} - \varepsilon_{xy} \Big) = \frac{1}{2} ( 1 - v ) - (1 - v) \ewav = ( 1 - v ) \Big( \frac{1}{2} - \ewav \Big)
\end{equation*}
for $\ewav$ defined in Proposition~\ref{prop:ewav-eigenvalue}. Combining Propositions~\ref{prop:ewav-eigenvalue} and~\ref{prop:centre-bound} leads to
\begin{equation*}
\ave[\big]{ ( \Lambda_{A}^{1} \otimes \Lambda_{B}^{1} )( \Psi ), \targetstate^{+} } \leq \frac{1}{2} + 596 \ewav.
\end{equation*}
Adding the two up immediately yields
\begin{equation*}
F \big( ( \Lambda_{A} \otimes \Lambda_{B} ) ( \rhoXYAB ) , \targetstate^{+} \big) \leq ( 1 - v ) \Big( \frac{1}{2} - \ewav \Big) + v \Big( \frac{1}{2} + 596 \ewav \Big) = \frac{1}{2} + ( 597 v - 1 ) \ewav = \frac{1}{2}.
\end{equation*}
\end{proof}
The value $p_{11} = 1 / 597$ plugged into Eq.~\eqref{eq:chsh-violation} gives the CHSH violation of $\beta \approx \thrviol$.
\subsection{Proof details}
\label{app:proof-details}
In this section we prove Propositions~\ref{prop:ewav-eigenvalue} and \ref{prop:centre-bound} used in the main argument. To do that we first need to prove three auxiliary lemmas.

The first lemma is a triangle-type inequality for the inner product of (finite-dimensional) density matrices.
\begin{lemma}
\label{lem:inner-product-triangle}
For finite-dimensional density matrices $\rho_{0}, \rho_{1}$ and $\sigma$ we always have
\begin{equation*}
\ave{ \rho_{0}, \rho_{1} } \geq 2 \big( \ave{ \rho_{0}, \sigma } + \ave{ \rho_{1}, \sigma } \big) - 3.
\end{equation*}
In particular, if
\begin{align*}
\ave{ \rho_{0}, \sigma } &\geq 1 - \delta_{0},\\
\ave{ \rho_{1}, \sigma } &\geq 1 - \delta_{1},
\end{align*}
then
\begin{equation*}
\ave{ \rho_{0}, \rho_{1} } \geq 1 - 2 ( \delta_{0} + \delta_{1} ).
\end{equation*}
\end{lemma}
\begin{proof}
The triangle inequality for the Schatten 2-norm (the Frobenius norm) implies that
\begin{equation*}
\norm{ \rho_{0} - \rho_{1} }_{2} \leq \norm{ \rho_{0} - \sigma }_{2} + \norm{ \sigma - \rho_{1} }_{2},
\end{equation*}
which can be written as
\begin{equation*}
\sqrt{ \ave{ \rho_{0}, \rho_{0} } + \ave{ \rho_{1}, \rho_{1} } - 2 \ave{ \rho_{0}, \rho_{1} } } \leq \sqrt{ \ave{ \rho_{0}, \rho_{0} } + \ave{ \sigma, \sigma } - 2 \ave{ \rho_{0}, \sigma } } + \sqrt{ \ave{ \rho_{1}, \rho_{1} } + \ave{ \sigma, \sigma } - 2 \ave{ \rho_{1}, \sigma } }.
\end{equation*}
Since both sides are non-negative, we can square the inequality to obtain
\begin{equation*}
- \ave{ \rho_{0}, \rho_{1} } \leq \ave{ \sigma, \sigma } - \ave{ \rho_{0}, \sigma } - \ave{ \rho_{1}, \sigma } + \sqrt{ \big( \ave{ \rho_{0}, \rho_{0} } + \ave{ \sigma, \sigma } - 2 \ave{ \rho_{0}, \sigma } \big) \big( \ave{ \rho_{1}, \rho_{1} } + \ave{ \sigma, \sigma } - 2 \ave{ \rho_{1}, \sigma } \big) }.
\end{equation*}
The fact that for an arbitrary density matrix $\tau$ we have $\ave{ \tau, \tau } \leq 1$ gives
\begin{equation*}
- \ave{ \rho_{0}, \rho_{1} } \leq 1 - \ave{ \rho_{0}, \sigma } - \ave{ \rho_{1}, \sigma } + 2 \sqrt{ \big( 1 - \ave{ \rho_{0}, \sigma } \big) \big( 1 - \ave{ \rho_{1}, \sigma } \big) }.
\end{equation*}
We bound the last term using the mean inequality $\sqrt{ a b } \leq ( a + b )/2$ which leads to
\begin{equation*}
- \ave{ \rho_{0}, \rho_{1} } \leq 3 - 2 \big( \ave{ \rho_{0}, \sigma } + \ave{ \rho_{1}, \sigma } \big).
\end{equation*}
\end{proof}
The second lemma formalises the intuition that an arbitrary channel acting jointly on a classical and quantum register can be replaced by a channel that reads the classical register and acts on the quantum register accordingly.
\begin{lemma}
\label{lem:classical-quantum-channel}
Let $\sH_{C}, \sH_{Q}$ and $\sH_{A}$ be Hilbert spaces of dimensions $d_{C}, d_{Q}$ and $d_{A}$, respectively. Let $\{ \ket{e_{j}} \}_{j = 1}^{ d_{C} }$ be an orthonormal basis of $\sH_{C}$ and we say that $R_{CQ}$ is a classical-quantum operator acting on $\sH_{C} \otimes \sH_{Q}$ if it can be written as
\begin{equation}
\label{eq:classical-quantum-operator}
R_{CQ} = \sum_{j} \ketbraq{e_{j}} \otimes S_{j}
\end{equation}
for some linear operators $S_{j} \in \cL( \sH_{Q} )$. Then, for an arbitrary channel $\Lambda : \cL( \sH_{C} \otimes \sH_{Q} ) \to \cL( \sH_{A} )$ there exists a collection of $d_{C}$ channels $\Lambda_{j} : \cL( \sH_{Q} ) \to \cL( \sH_{A} )$ such that for all operators of the form~\eqref{eq:classical-quantum-operator} we have
\begin{equation}
\label{eq:channel-equality}
\Lambda( R_{CQ} ) = \sum_{j} \Lambda_{j}( S_{j} ).
\end{equation}
\end{lemma}
\begin{proof}
We define the channel $\Lambda_{j}$ through its action on an arbitrary operator $X \in \cL( \sH_{ Q } )$. Let
\begin{equation*}
\Lambda_{j}( X ) := \Lambda \big( \ketbraq{ e_{j} } \otimes X \big),
\end{equation*}
which ensures that $\Lambda_{j}$ is completely positive and trace-preserving. The equality~\eqref{eq:channel-equality} holds by construction.
\end{proof}
The last lemma shows that if a channel maps the maximally mixed state to a state which is close to being pure, then this channel must contract all the Pauli observables.
\begin{lemma}
\label{lem:spectrum-bound}
Let $\Lambda$ be a qubit quantum channel, let
\begin{equation*}
\omega := \Lambda \bigg( \frac{ \unit_2 }{2} \bigg) \, ,
\end{equation*}
and suppose that $\spec( \omega ) = \{ \lambda, 1 - \lambda \}$ for $\lambda \in [0, 1/2]$. Let $\Gamma$ be a $2 \times 2$ Hermitian operator satisfying $\Gamma^{2} = \unit$ and $\tr \Gamma = 0$. Then,
\begin{equation*}
- 2 \sqrt{\lambda} \, \unit_{2} \leq \Lambda( \Gamma) \leq 2 \sqrt{\lambda} \, \unit_{2}.
\end{equation*}
\end{lemma}
\begin{proof}
Since the quantum channel is a positive map, we have $\Lambda ( \unit_2 \pm \Gamma ) \geq 0$ or, equivalently $ - 2 \omega \leq \Lambda( \Gamma ) \leq 2\omega$. We start by writing both operators in the eigenbasis of $\omega$
\begin{equation*}
\omega =
\left(
\begin{array}{cc}
\lambda &\\
& 1 - \lambda
\end{array}
\right) \nbox{and}
\Lambda(\Gamma) =
\left(
\begin{array}{cc}
t & y\\
y^{*} & -t
\end{array}
\right)
\end{equation*}
for some $t \in \amsbb{R}$ and $y \in \amsbb{C}$. Note that we have implicitly used the fact that $\Lambda(\Gamma)$ is Hermitian and traceless. The condition $\Lambda(\Gamma) \geq - 2 \omega$ reads
\begin{equation*}
\left(
\begin{array}{cc}
2 \lambda + t & y\\
y^{*} & 2 - 2 \lambda - t
\end{array}
\right) \geq 0
\end{equation*}
and implies that
\begin{equation*}
( 2 \lambda + t ) ( 2 - 2 \lambda - t ) - \abs{y}^2 \geq 0.
\end{equation*}
Similarly, the condition $\Lambda(\Gamma) \leq 2 \omega$ leads to
\begin{equation*}
( 2 \lambda - t ) ( 2 - 2 \lambda + t ) - \abs{y}^2 \geq 0.
\end{equation*}
Adding up these two conditions gives
\begin{equation*}
t^{2} + \abs{y}^{2} \leq 4 \lambda ( 1 - \lambda) \leq 4 \lambda.
\end{equation*}
As the eigenvalues of $\Lambda(\Gamma)$ are easily seen to be $\pm \sqrt{ t^{2} + \abs{y}^{2} }$, the claim follows directly from the last inequality.
\end{proof}
Equipped with these three auxiliary lemmas we are ready to tackle the two propositions used in the main argument.
\begin{prop:counterexample-1}
\propewaveigenvalue
\end{prop:counterexample-1}
\begin{proof}
The proof consists of three steps. We first consider the four corner points, i.e.~$(x, y) \in \{ (0, 0), (0, 2), (2, 0), (2, 2) \}$ and show that the channels $\Lambda_{A}^{0}$ and $\Lambda_{B}^{0}$ map the entire Bloch sphere to a small region close to the boundary. In the second step we consider the points $(x, y) \in \{ (0, 1), (1, 0) \}$ to show that the channels $\Lambda_{A}^{1}$ and $\Lambda_{B}^{1}$ have the same property. In the last step we compute an explicit bound on the eigenvalues of $\omega_{A}$ and $\omega_{B}$.

For $b \in \{0, 1\}$ and $x, y \in \{0, 1, 2\}$ define
\begin{align*}
\sigma_{b}^{x} &:= \Lambda_{A}^{x} ( \ketbraq{b} ),\\
\tau_{b}^{y} &:= \big[ \Lambda_{B}^{y} ( \ketbraq{b} ) \big]\tran,
\end{align*}
which implies that
\begin{equation*}
\ave[\big]{ ( \Lambda_{A}^{x} \otimes \Lambda_{B}^{y} )( \ketbraq{b} \otimes \ketbraq{b'} ), \targetstate^{+} } = \ave[\big]{ \Lambda_{A}^{x} ( \ketbraq{b} ) \otimes \Lambda_{B}^{y} ( \ketbraq{b'} ), \targetstate^{+} } = \frac{1}{2} \ave{ \sigma_{b}^{x}, \tau_{b'}^{y} }.
\end{equation*}
Therefore, Eq.~\eqref{eq:epsilonxy} imposes constraints on the inner products between the operators $\sigma_{b}^{x}$ and $\tau_{b}^{y}$. Considering points $(x, y) = (0, 0), (0, 2), (2, 0), (2, 2)$ gives
\begin{align}
\label{eq:epsilon00}
\ave{ \sigma_{1}^{0}, \tau_{1}^{0} } &= 1 - 2\varepsilon_{00},\\
\label{eq:epsilon02}
\ave{ \sigma_{0}^{0}, \tau_{0}^{2} } + \ave{ \sigma_{1}^{0}, \tau_{1}^{2} } &= 2 - 4\varepsilon_{02},\\
\label{eq:epsilon20}
\ave{ \sigma_{0}^{2}, \tau_{0}^{0} } + \ave{ \sigma_{1}^{2}, \tau_{1}^{0} } &= 2 - 4\varepsilon_{20},\\
\label{eq:epsilon22}
\ave{ \sigma_{0}^{2}, \tau_{1}^{2} } + \ave{ \sigma_{1}^{2}, \tau_{0}^{2} } &= 2 - 4\varepsilon_{22}.
\end{align}
Plugging the upper bound $\ave{ \sigma_{b}^{x}, \tau_{b'}^{y} } \leq 1$ into Eq.~\eqref{eq:epsilon02} immediately gives
\begin{equation*}
\ave{ \sigma_{1}^{0}, \tau_{1}^{2} } \geq 1 - 4\varepsilon_{02},
\end{equation*}
which combined with Eq.~\eqref{eq:epsilon00} by Lemma~\ref{lem:inner-product-triangle} gives
\begin{equation}
\label{eq:tau1}
\ave{ \tau_{1}^{0}, \tau_{1}^{2} } \geq 1 - 4 ( \varepsilon_{00} + 2 \varepsilon_{02} ).
\end{equation}
Similarly, Eqs.~\eqref{eq:epsilon20} and~\eqref{eq:epsilon22} imply
\begin{align*}
\ave{ \sigma_{0}^{2}, \tau_{0}^{0} } &\geq 1 - 4 \varepsilon_{20},\\
\ave{ \sigma_{0}^{2}, \tau_{1}^{2} } &\geq 1 - 4 \varepsilon_{22},
\end{align*}
which gives
\begin{equation*}
\ave{ \tau_{0}^{0}, \tau_{1}^{2} } \geq 1 - 8 ( \varepsilon_{20} + \varepsilon_{22} ).
\end{equation*}
Combining this with Eq.~\eqref{eq:tau1} gives
\begin{equation}
\label{eq:taus}
\ave{ \tau_{0}^{0}, \tau_{1}^{0} } \geq 1 - 8 \big[ \varepsilon_{00} + 2 ( \varepsilon_{02} + \varepsilon_{20} + \varepsilon_{22}) \big],
\end{equation}
which concludes the first step of the proof. This lower bound implies that the states $\tau_{0}^{0}$ and $\tau_{1}^{0}$ are close to each other and, moreover, that they are close to being pure. Since these two states result from applying the channel $\Lambda_{B}^{0}$ to two pure orthogonal states, we conclude that the channel must shrink the entire Bloch sphere to a small region close to the boundary.

Considering the point $(x, y) = (1, 0)$ gives
\begin{equation*}
\ave{ \sigma_{0}^{1}, \tau_{0}^{0} } + \ave{ \sigma_{1}^{1}, \tau_{1}^{0} } = 2 - 4\varepsilon_{10}.
\end{equation*}
Define $a, b \geq 0$ such that
\begin{align}
\label{eq:sigma-tau1}
\ave{ \sigma_{0}^{1}, \tau_{0}^{0} } &= 1 - a,\\
\label{eq:sigma-tau2}
\ave{ \sigma_{1}^{1}, \tau_{1}^{0} } &= 1 - b,
\end{align}
which implies that $a + b = 4\varepsilon_{10}$. Applying the inner-product inequality proven in Lemma~\ref{lem:inner-product-triangle} to Eqs.~\eqref{eq:taus}, \eqref{eq:sigma-tau1} and \eqref{eq:sigma-tau2} gives
\begin{equation*}
\ave{ \sigma_{0}^{1}, \sigma_{1}^{1} } \geq 1 - 32 \big[ \varepsilon_{00} + 2 ( \varepsilon_{02} + \varepsilon_{20} + \varepsilon_{22}) \big] - 4 a - 2 b
\end{equation*}
or
\begin{equation*}
\ave{ \sigma_{0}^{1}, \sigma_{1}^{1} } \geq 1 - 32 \big[ \varepsilon_{00} + 2 ( \varepsilon_{02} + \varepsilon_{20} + \varepsilon_{22}) \big] - 2 a - 4 b
\end{equation*}
depending on the order. Averaging over these two bounds gives
\begin{align*}
\ave{ \sigma_{0}^{1}, \sigma_{1}^{1} } &\geq 1 - 32 \big[ \varepsilon_{00} + 2 ( \varepsilon_{02} + \varepsilon_{20} + \varepsilon_{22}) \big] - 3 a - 3 b\\
&= 1 - 32 \big[ \varepsilon_{00} + 2 ( \varepsilon_{02} + \varepsilon_{20} + \varepsilon_{22}) \big] - 12 \varepsilon_{10}\\
&= 1 - \finexp{10},
\end{align*}
which concludes the second step of the proof.

The density matrix $\omega_{A}$ defined in the proposition is given by
\begin{equation*}
\omega_{A} = \Lambda_{A}^{1} \bigg( \frac{ \unit_2 }{2} \bigg) = \frac{1}{2} ( \sigma_{0}^{1} + \sigma_{1}^{1} ).
\end{equation*}
Clearly, $\tr \omega_{A} = 1$ and
\begin{align*}
\tr \omega_{A}^{2} &= \frac{1}{4} \big[ \ave{ \sigma_{0}^{1}, \sigma_{0}^{1} } + \ave{ \sigma_{1}^{1}, \sigma_{1}^{1} } + 2 \ave{ \sigma_{0}^{1}, \sigma_{1}^{1} } \big]\\
&= \frac{1}{4} \big[ \ave{ \sigma_{0}^{1}, \sigma_{0}^{1} } + \ave{ \sigma_{1}^{1}, \sigma_{1}^{1} } - 2 \ave{ \sigma_{0}^{1}, \sigma_{1}^{1} } \big] + \ave{ \sigma_{0}^{1}, \sigma_{1}^{1} }\\
&= \frac{1}{4} \big[ \ave{ \sigma_{0}^{1} - \sigma_{1}^{1}, \sigma_{0}^{1} - \sigma_{1}^{1} } \big] + \ave{ \sigma_{0}^{1}, \sigma_{1}^{1} } \geq \ave{ \sigma_{0}^{1}, \sigma_{1}^{1} }.
\end{align*}
We take advantage of the fact that for $2 \times 2$ Hermitian matrices we have $ [\tr(M)]^{2} = \tr ( M^{2} ) + 2 \det(M)$. If $\lambda$ is the smaller eigenvalue of $\omega_{A}$, then
\begin{align*}
\lambda \leq 2 \lambda ( 1 - \lambda ) &= 2 \det( \omega_{A} ) = [\tr(\omega_{A})]^{2} - \tr ( \omega_{A}^{2} ) = 1 - \tr( \omega_{A}^{2} )\\
&\leq 1 - \ave{ \sigma_{0}^{1}, \sigma_{1}^{1} } \leq \finexp{10},
\end{align*}
which concludes the last step of the proof of the first statement. The proof of the second statement is essentially the same. From Eqs.~\eqref{eq:epsilon00} and \eqref{eq:epsilon20} we obtain
\begin{equation*}
\ave{ \sigma_{1}^{0}, \sigma_{1}^{2} } \geq 1 - 4 ( \varepsilon_{00} + 2 \varepsilon_{20} ),
\end{equation*}
whereas Eqs.~\eqref{eq:epsilon02} and~\eqref{eq:epsilon22} imply
\begin{equation*}
\ave{ \sigma_{0}^{0}, \sigma_{1}^{2} } \geq 1 - 8 ( \varepsilon_{02} + \varepsilon_{22} ).
\end{equation*}
Combining these yields
\begin{equation*}
\ave{ \sigma_{0}^{0}, \sigma_{1}^{0} } \geq 1 - 8 \big[ \varepsilon_{00} + 2 ( \varepsilon_{02} + \varepsilon_{20} + \varepsilon_{22}) \big]
\end{equation*}
and by adding the point $(x, y) = (0, 1)$ we arrive at
\begin{equation*}
\ave{ \tau_{0}^{1}, \tau_{1}^{1} } \geq 1 - \finexp{01}.
\end{equation*}
Finally, we note that $\omega_{B}\tran = ( \tau_{0}^{1} + \tau_{1}^{1} ) / 2$, but since the transpose does not affect the spectrum, the final calculation is precisely the same.
\end{proof}
\begin{prop:counterexample-2}
\propcentrebound
\end{prop:counterexample-2}
\begin{proof}
Since the statement is invariant under local unitaries, we can without loss of generality assume that $\Psi_{1}$ is the usual maximally entangled state, i.e.
\begin{equation*}
\Psi_{1} = \frac{1}{4} \Big( \unit \otimes \unit + \X \otimes \X - \Y \otimes \Y + \Z \otimes \Z \Big).
\end{equation*}
Note that $\Psi_{1}$ can be written as
\begin{equation*}
\Psi_{1} = \tau + \frac{1}{4} ( \X \otimes \X + \Z \otimes \Z \Big),
\end{equation*}
where $\tau = ( \unit \otimes \unit - \Y \otimes \Y ) / 4$. Since $\tau$ is a separable state, we have
\begin{equation*}
\ave[\big]{ ( \Lambda_{A} \otimes \Lambda_{B} )( \tau ), \Psi_{2} } \leq \frac{1}{2}.
\end{equation*}
To bound the other two terms we use Lemma~\ref{lem:spectrum-bound}, which in particular implies that
\begin{equation*}
\Lambda_{A}(\X) \otimes \Lambda_{B}(\X) \leq 4 \lambda \, \unit_{4}.
\end{equation*}
Therefore,
\begin{equation*}
\ave[\big]{ \Lambda_{A}(\X) \otimes \Lambda_{B}(\X), \Psi_{2} } \leq 4 \lambda.
\end{equation*}
The same argument applied to $\Lambda_{A}(\Z) \otimes \Lambda_{B}(\Z)$ concludes the proof.
\end{proof}
\end{document}